\documentclass[12pt]{article}
\usepackage{amsmath,amsthm,amsfonts,amssymb}
\usepackage{mathrsfs,bm}
\usepackage{graphicx,amscd,epsfig,psfrag,verbatim}

\usepackage{hyperref}
\hypersetup{
    colorlinks = true,
    citecolor = {blue},
    linkcolor = {blue},
}
\usepackage{natbib}
\usepackage{booktabs}

\usepackage{algorithm}
\usepackage{algpseudocode}

\usepackage[a4paper, total={7.2in, 9.5in}]{geometry}

\usepackage{xr}
\newtheorem{theorem}{Theorem}

\newtheorem{lemma}[theorem]{Lemma}
\newtheorem{corollary}[theorem]{Corollary}

\newcommand{\numberthis}{\addtocounter{equation}{1}\tag{\theequation}}

\usepackage{amsfonts,amssymb,mathrsfs,bm,bbm,dsfont}

\renewcommand{\d}{\text{\rm\,d}}

\def\R{\mathbb{R}}

\def\N{\mathbb{N}}
\def\E{\mathbb{E}}
\def\P{\mathbb{P}}

\def\eps{\varepsilon}
\def\del{\delta}

\def\L{\Lambda}
\def\la{\lambda}
\def\t{\theta}
\def\d{\mathrm{d}}
\def\a{\alpha}
\def\b{\beta}
\def\S{\Sigma}
\def\calL{\mathcal{L}}

\def\s{\sigma}

\def\tr{\mathrm{Tr}}

\def\ind{{\mathds{1}}}
\def\cC{\mathcal{C}}

\def\sgn{\mathrm{sgn}}

\def\cG{\mathcal{G}}

\def\X{\mathbb{X}}
\def\calC{\mathcal{C}}
\def\cmax{\mathcal{C}_{\max}}
\def\ind{\mathbbm{1}}

\def\low{\mathrm{low}}

\def\slow{\mathrm{Spec}_{\mathrm{low}}}
\def\spec{\mathrm{Spec}}
\def\tflow{t_{\mathrm{flow}}}
\def\Ns{N_{\mathrm{step}}}
\def\deg{\mathrm{deg}}

\def\md{\mathfrak{m}}
\def\snd{\mathfrak{d}}
\def\cP{\mathcal{P}}

\newcommand{\bzero}{\mathbf{0}}
\newcommand{\bone}{\mathbf{1}}

\renewcommand{\l}[0]{\left }
\renewcommand{\r}[0]{\right}
\newcommand{\op}{\mathrm{op}}
\newcommand{\RE}{\mathrm{RE}}

\usepackage{xcolor}

\title{\bf Learning with latent group sparsity via heat flow dynamics on networks}
\author{Subhroshekhar Ghosh \\
    Department of Mathematics \\
    National University of Singapore \\
    10 Lower Kent Ridge Road \\
    Singapore 119076 \\
    E-mail: \texttt{subhrowork@gmail.com} \\
    \url{https://subhro-ghosh.github.io}
    \and
    Soumendu Sundar Mukherjee \\
    Interdisciplinary Statistical Research Unit \\
    Indian Statistical Institute \\
    203 Barrackpore Trunk Road \\
    Kolkata, India 700108 \\
    E-mail: \texttt{soumendu041@gmail.com} \\
    \url{https://soumendu041.gitlab.io}
}

\date{}

\begin{document}
\maketitle

\begin{abstract}
    Group or cluster structure on explanatory variables in machine learning problems is a very general phenomenon, which has attracted broad interest from practitioners and theoreticians alike. In this work we contribute an approach to learning under such group structure, that does not require prior information on the group identities. Our paradigm is motivated by the Laplacian geometry of an underlying network with a related community structure, and proceeds by directly incorporating this into a penalty that is effectively computed via a heat flow-based local network dynamics. In fact, we demonstrate a procedure to construct such a network based on the available data. Notably, we dispense with computationally intensive pre-processing involving clustering of variables, spectral or otherwise. Our technique is underpinned by rigorous theorems that guarantee its effective performance and provide bounds on its sample complexity. In particular, in a wide range of settings, it provably suffices to run the heat flow dynamics for time that is only logarithmic in the problem dimensions. We explore in detail the interfaces of our approach with key statistical physics models in network science, such as the Gaussian Free Field and the Stochastic Block Model. We validate our approach by successful applications to real-world data from a wide array of application domains, including computer science, genetics, climatology and economics. Our work raises the possibility of applying similar diffusion-based techniques to classical learning tasks, exploiting the interplay between geometric, dynamical and stochastic structures underlying the data.
\end{abstract}


\section{Introduction}\label{sec:intro}
The understanding and analysis of data with complex structure is a leitmotif of modern science and technology.  The spectacular growth in the capacity and computational means to process gigantic volumes of data has motivated the development of novel analytical paradigms in recent years.  A common theme that characterises many of these approaches is that they seek to incorporate the growing complexity that is inherent in such massive data sets \citep{national2013frontiers,marx2013big}.  

The intrinsic structure in data can manifest itself in various forms.  These range from the almost ubiquitous scenario of sparsity in an appropriate basis,  such as in compressive sensing and low-rank estimation problems \citep{achlioptas2007fast,grasedyck2013literature,chang2000adaptive,foucart2013invitation,berthet2013optimal},  to algebraic constraints imposed by physical considerations,  such as  symmetries under rigid motions that are inherent in problems of cryo-electron microscopy \citep{cheng2015primer,singer2018mathematics,hadani2011representation,perry2019sample,fan2020likelihood,ghosh2021multi}. In yet other instances,  constraints may be stochastic in nature,  pertaining to the statistical dependency structures that  characterise  the model \citep{mezard2009information,ros2019complex,ghosh2020fractal,lauritzen2019maximum,li2016fast,lavancier2015determinantal,ghosh2020gaussian,bardenet2021determinantal}.

A significant structural feature that arises in  many application scenarios is the clustering, or grouping, of some of the explanatory variables into a relatively limited number of categories, with the understanding that the quantities in the same category are strongly dependent on each other. A typical scenario is the existence of deterministic relationships governing the values of the variables in the same group.  Further, it is often the case that only a few of these groups  contribute meaningfully to the experimental observations, with the remaining variables being redundant or uninformative for predictive purposes. An important use-case of such structure is that of high dimensional supervised learning, where the explanatory variables are often clustered via natural constraints. For instance, meteorological measurements in spatially adjacent locations are likely to be highly correlated. Similarly, frequency of occurrence of certain words or phrases in spam emails are likely to be highly correlated.

A  different setting in which a group structure on variables plays a significant role is that of community detection problems, where groups or clusters pertain to connectivity patterns in an underlying network. A typical example is that of a social network, where connectivity corresponds to friendship or acquaintance; yet another is that of collaboration network among scientists \citep{fortunato2010community,fortunato2007resolution,reichardt2006statistical}.  Due to obvious practical ramifications,  this area has witnessed intense research activity in recent years,  a significant achievement of which is the extensive theory of \textit{Stochastic Block Models} (\textit{abbrv.} SBM) \citep{abbe2017community,abbe2015exact,bandeira2016low,goldenberg2010survey,holland1983stochastic,karrer2011stochastic}. In general, incorporating network geometry into standard statistical learning problems has been an area of recent interest \citep{hallac2015network, li2020graph, li2019prediction, li2020high}.

The network structure brings into  focus the geometric perspective on the clustering phenomenon,  that is underpinned by the metric  induced by the weighted graph distance in the network \citep{von2007tutorial}.  Intertwined with such geometry is the canonical dynamics  associated to it --  in a very general Riemannian geometric setting,  the metric structure gives rise to a Laplacian operator, which in turn serves as the generator of the so-called \textit{heat flow} dynamics on the underlying space.  These correspondences are classical in metric geometry and harmonic analysis
\citep{rosenberg1997laplacian,jost2008riemannian}.
 
Associated with the Laplacian geometry and heat flow dynamics is the canonical model in statistical physics referred to as the \textit{Gaussian Free Field} (\textit{abbrv.} GFF) (c.f.  \cite{sheffield2007gaussian,berestycki2015introduction,friedli2017statistical}), which has also emerged to be of independent interest as an important instance of \textit{Gaussian graphical models}  \citep{zhu2003semi,zhu2003combining,ma2013sigma,kelner2019learning,rasmussen2003gaussian}.   GFF-s complete the above picture from a statistical point of view,  by embedding the stochastic dependency structure of Gaussian random variables in the setup of the geometric structure and dynamical properties of a weighted networks.
 
In the present work, we bring together these disparate strands  into synergy -- clustering phenomena on variables or predictors in one direction, those in SBM-type network models in another, and the  geometry of and dynamics on weighted networks in yet a third direction, along with their statistical physical implications. Leveraging their interplay,  we obtain an algorithm to perform effective regression analysis in both high and low dimensional setup for variables with a \textit{latent group structure},  using  \textit{limited and local access} to  underlying network dependencies.  In a more general setup, when an underlying graph may not be explicit in the problem description, we demonstrate a procedure to construct such a network based on the available data.  Substantiated by rigorous mathematical analysis and robust experimental performance,  our approach may be seen to outperform more classical approaches for group structured data.  From an algorithmic perspective,  our methodology alludes to interesting connections the with so-called \textit{diffusion mapping} techniques, which have been effective as dimension reduction tools \citep{coifman2005geometric, coifman2006diffusion}.    

More generally,  our approach opens the avenue to applications of similar diffusion-based techniques to classical statistical and data analytical problems, that are generally static in nature. The inherently local nature of the heat flow and related diffusion dynamics enables us to solve the relevant constrained optimization problems while being oblivious to the global geometry of the graph.  In addition to economies of computational resources,  such locality is of significance with regard to questions of privacy in data analysis, a problem that is gaining increasing salience in today's hyper-networked world. 
 
\section{Lasso and its derivatives}

Variable selection is a classic problem in statistics, which has become all the more important in the present age with the routine availability of large scientific datasets with measurements on tens of thousands of variables. Sparsity has become a key methodological instrument for meaningful inference from such ``high-dimensional'' datasets. Parsimonious models are easier to interpret and more resistant to overfitting. The Least Absolute Shrinkage and Selection Operator (abbrv. \emph{lasso}) \citep{tibshirani1996regression}, which employs a sparsity inducing $\ell_1$-penalty, is perhaps the most prominent method of variable selection.

A major problem with vanilla lasso is that it treats all variables equally. Thus, when there are natural groups in the variables, some variables in a group can get kicked off the model with other members still included. The \emph{group lasso} penalty \citep{yuan2006model} aims to solve this issue.  Consider a supervised learning problem with $p$ predictors and corresponding parameter $\beta \in \R^p$. Denoting the groups by $\cC_1$, \ldots, $\cC_k$, the group lasso penalty uses a weighted $\ell_1$ norm of the groupwise $\ell_2$ norms:
\[
    \mathrm{GL}(\beta) = \sum_{\ell = 1}^k \sqrt{|\cC_\ell|} \,\, \|\beta_{\cC_\ell}\|_2.
\]
As a consequence, variables in a group exit the model together.

The group lasso penalty requires the groups to be known in advance.  However, in many practical scenarios, the group information is a priori unknown to the statistician. There has been some work to address this issue. The \emph{cluster representative lasso} (CRL) and \emph{cluster group lasso} (CGL) algorithms of \cite{buhlmann2013correlated}  perform an initial clustering of variables into groups and then use the estimated clusters for grouped variable selection.

In this paper,  we directly construct a penalty that automatically selects variables in groups without any prior group information. Notably,  we dispense with the elaborate pre-processing step involving clustering of the variables,  spectral or otherwise,  which can be computationally intensive.

\section{Laplacian geometry of graphs and the heat flow penalty}\label{sec:setup}
Suppose that we have explanatory variables $X_1, \ldots, X_p$,  whose group structure is captured by a graph $G=(V,E)$ on $|V|=p$ vertices,  each vertex corresponding to an $X_i$.  In the simplest scenario,  the groups of variables would correspond to the connected components $\{\calC_i\}_{i=1}^k$ of $G$.  More generally,  we consider the nodes of $G$ to be a union of subsets  $\{\calC_i\}_{i=1}^k$ (corresponding to the variable groups),  such that the $\calC_i$-s are relatively densely intra-connected,  but the inter-connections across the $\calC_i$-s are relatively sparse. This is similar in flavour to the problem of multi-way partitioning of graphs,  which has attracted considerable interest over the years \citep{lee2014multiway}.

Let $L$ denote the Laplacian of $G$.  We introduce the penalty
\begin{equation}\label{eq:pen-hflasoo}
    \Lambda_t(\beta) = \langle \Phi(e^{-tL}(\beta \odot \beta)), \mathbf{1} \rangle,
\end{equation}
where $\odot$ denotes Hadamard/elementwise product of vectors, $\Phi(\beta) = (\sqrt{|\beta_1|}, \ldots, \sqrt{|\beta_p|})^\top$,  and $\mathbf{1}$ is the all ones vector in $\R^p$.  
As we show in Lemma 1 and Corollary 2 in the appendix, in the limit of $t \to \infty$, the quantity $\Lambda_t(\beta)$ approaches the classical group lasso penalty $\mathrm{GL}(\beta)$ in the setting where the components $\{\calC_i\}_{i=1}^k$ are fully disconnected from each other in the graph $G$.  See Figure~\ref{fig:pen-balls} for a demonstration of this convergence. 

Note that $\Lambda_t$ directly incorporates the Laplacian in terms of the heat flow operator $e^{-tL}$ on the underlying graph $G$. Unlike the group lasso penalty, $\Lambda_t$ is non-convex which is a potential problem vis-a-vis  optimization. However, using the connection with heat flow, we can calculate $\Lambda_t$ in an efficient manner via random walks on the graph $G$. 

We consider the penalised supervised learning problem 
\begin{equation} \label{eq:pen_loss}
    F_{t, \lambda}(\beta;X,y) = \calL(\beta;X,y)+ \lambda \Lambda_t(\beta),
\end{equation}
where $(X,y)$ denotes the training data,  and $\calL(\beta;X,y)$ is a suitable loss function defined with respect to the problem and computed using the training data at the parameter $\beta$.  We expect that for suitably large $t$, the solution of this optimization problem will be close to that of the classical group lasso problem. 

A fundamental example of this set-up is the penalised 
regression problem $\arg \min_{\beta} F_{t, \lambda}(\beta)$,  where
\[
    F_{t, \lambda}^{\text{reg}}(\beta) = \frac{1}{2n} \|y - X\beta\|_2^2 + \lambda \Lambda_t(\beta). 
\]
More generally,  an important family of examples  is accorded by the problem of penalised likelihood maximization,  where $\calL(\beta;X,y)$ is the negative log-likelihood of the observed data $(X,y)$ at the parameter value $\beta$.  Yet another instance is that of the so-called Huber's loss,  where the quadratic function of the $L_2$ norm in the regression set-up is replaced by a different convex function \citep{huber1992robust}.

\begin{figure}[!ht]
    \centering
    \begin{tabular}{cc}
        \includegraphics[scale = 0.30]{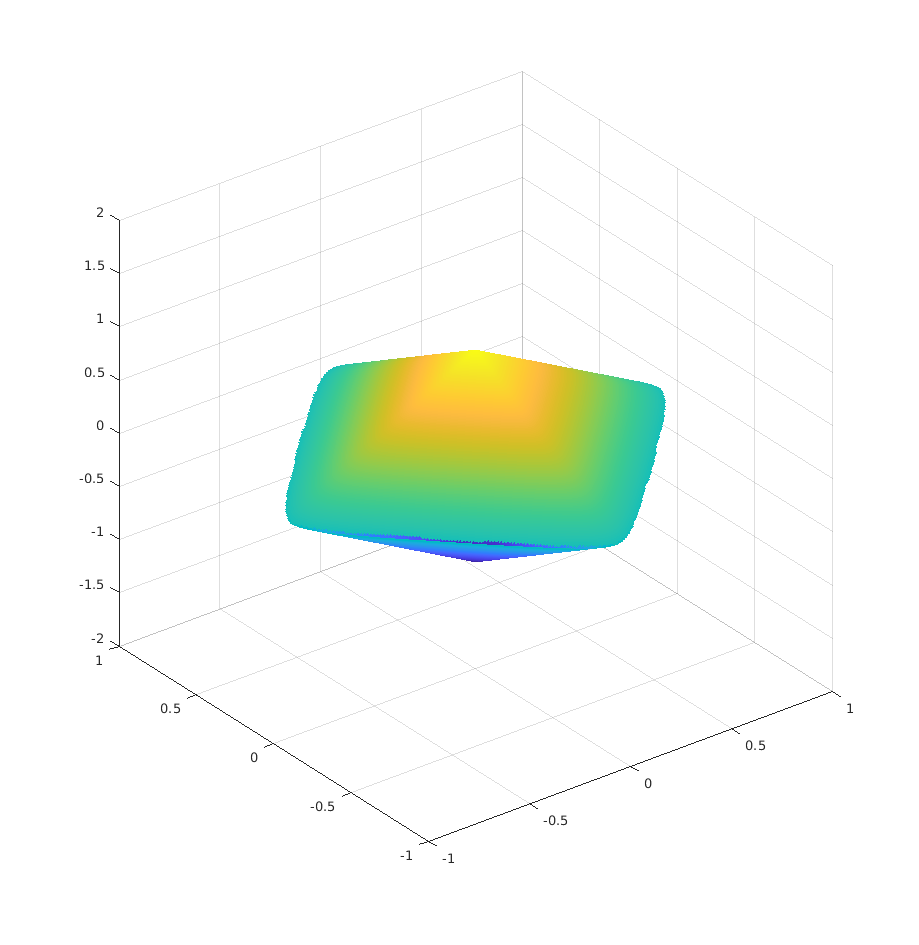} & \includegraphics[scale = 0.30]{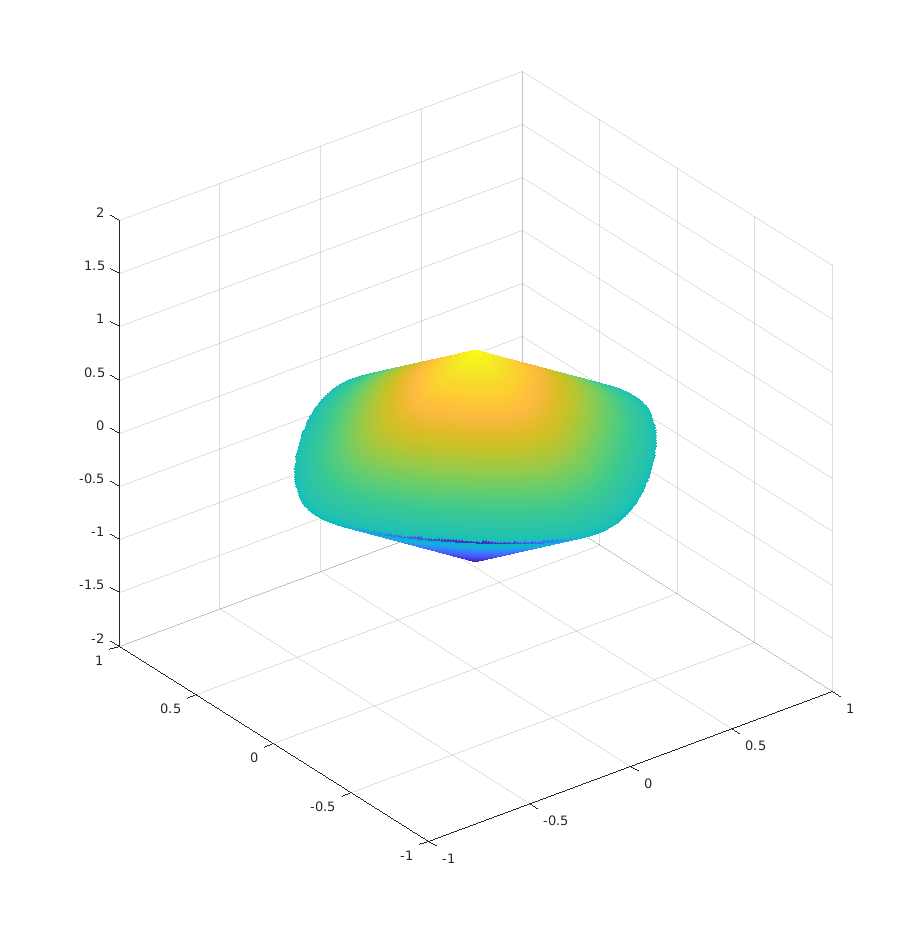} \\
        $t = 0.01$                         & $t = 0.1$ \\
         \includegraphics[scale = 0.30]{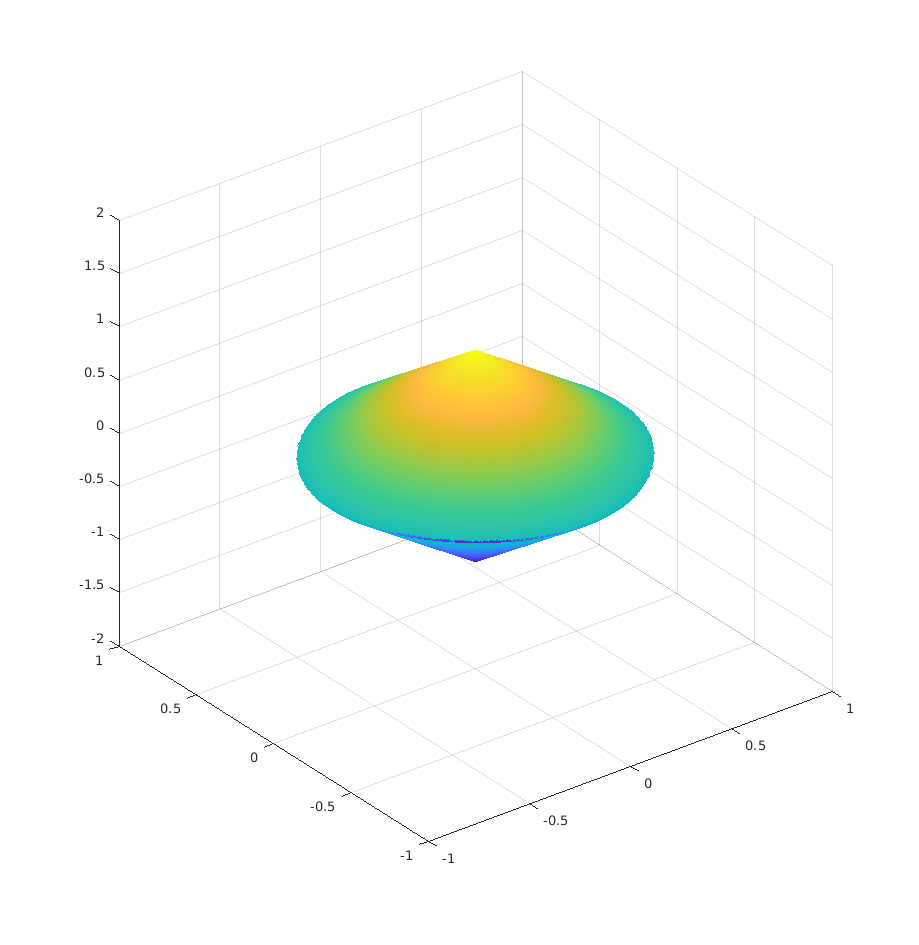} & \includegraphics[scale = 0.30]{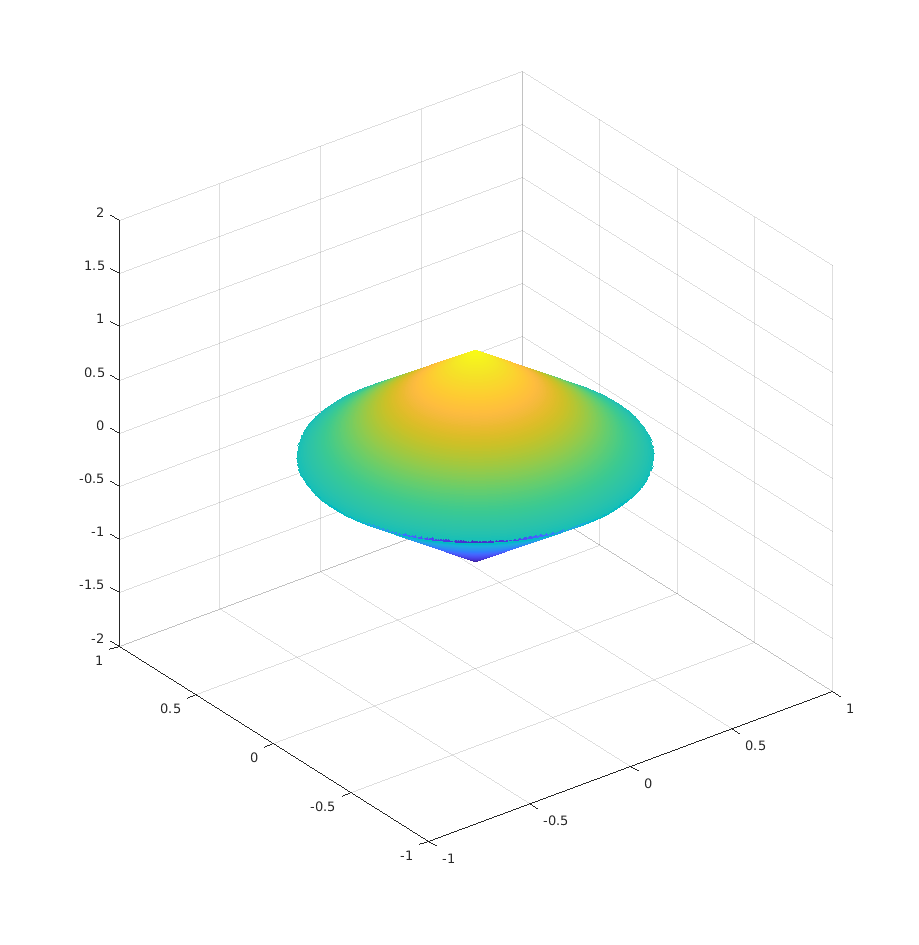} \\
         $t = 0.5$ & $t = 1$
    \end{tabular}
    \caption{Consider three variables with two groups $\{1, 2\}$ and $\{3\}$. In this display, we plot of the level set $\{\beta \mid \Lambda_t(\beta) \le 1\}$ for different values of $t$. The graph $G$ here is the union of an isolated vertex and an edge. The eigengap $\lambda_g = 2$.}
    \label{fig:pen-balls}
\end{figure}

\section{Local heat flow dynamics and an algorithm for implicit group sparsity} \label{sec:algo}

\algrenewcommand\algorithmicrequire{\textbf{Input:}}
\algrenewcommand\algorithmicensure{\textbf{Output:}}
\newcommand{\old}{\mathrm{old}}
\newcommand{\new}{\mathrm{new}}
\newcommand{\hflow}{\textsc{heatflow}}
\newcommand{\subgrad}{\mathrm{subgrad}}
\newcommand{\iter}{i}
\newcommand{\maxiter}{N}
\newcommand{\reldiff}{\mathrm{reldiff}}
\newcommand{\mean}{\mathrm{mean}}

\subsection{A heat-flow-based algorithm}
Our general algorithmic framework is underpinned by a heat-flow-based computation of vectors of the form $e^{-tL} f$ (or a subset of its co-ordinates) for a given $f \in \R^p$, where we view the latter as a function on the nodes of the graph $G$. In this setting, we observe that 
\[
    (e^{-tL} f)_i = \E[f_{X(t)} \mid X(0) = i], 
\]
where $X(t)$ is a continuous-time simple random walk on $G$, which is the canonical analogue of heat flow dynamics on in this setup. This indicates that if we start $B$ random walks $X^{(1)}, \ldots, X^{(B)}$ from the vertex $i$, then an estimate of $(e^{-tL} f)_i$ would be
\[
    \widehat{(e^{-tL} f)}_i = \frac{1}{B}\sum_{j = 1}^B f_{X^{(j)}(t)}.
\]

In Algorithm~\ref{alg:heatflow}, we describe the pseudo-code for simulating $B$ heat flows from each vertex, run till time $t$. Note that the for-loops can be easily paralellised, rendering this algorithm highly efficient. The running time (i.e. computational complexity) for Algorithm~\ref{alg:heatflow}, even without factoring in parallelisation, is only $O(pB\Ns)$, where $\Ns$ is the step count required for a single run of the heat flow. We demonstrate in \eqref{eq:eq:step_well_clustered} that $\Ns$ is typically $O(\max\{\log p, \log n\})$; as such the computational complexity for Algorithm~\ref{alg:heatflow} is only $O(B \cdot p \cdot \max\{\log p, \log n\})$.

The pseudo-code for approximating $e^{-tL} f$ (or a subset of its co-ordinates) using heat flow is presented in Algorithm~\ref{alg:heatflow_on_vector}. Finally,  the complete subgradient-based optimisation procedures are presented in Algorithms~\ref{alg:sd} and \ref{alg:block_cd}.

\begin{algorithm}[!ht]
    \caption{Simulate heat flow}\label{alg:heatflow}
    \begin{algorithmic}[1]
        \Require Graph $G$; time $t$ till which to run the heat flow dynamics; $B$, the number of heat flows generated per vertex
        \Ensure A $p \times B$ matrix $H$, with $H_{ij}$ storing the state of the $j$-th heat flow started at vertex $i$ at time $t$.
        \For{$i = 1, \ldots, p$}
        \For{$j = 1, \ldots, B$}
        \State $s \gets 0$
        \State $H_{ij} \gets i$
        \While{$s < t$}
        \State Generate $E \sim \mathrm{Exponential}(1)$
        \State $s \gets s + E$
        \If{$s < t$ \textbf{and} $\mathrm{degree}(H_{ij}) > 0$}
        \State $H_{ij} \gets $ a neighbour of $H_{ij}$, chosen at random 
        \EndIf

        \EndWhile
        \EndFor
        \EndFor
    \end{algorithmic}
\end{algorithm}

\begin{algorithm}[!ht]
    \caption{Heat flow on a vector $f$}\label{alg:heatflow_on_vector}
    \begin{algorithmic}[1]
        \Require Heat flow matrix $H$; vector $f$, indices $S$ at which value required
        \Ensure $g = \textsc{heatflow}(f, S)$, an estimate of $(e^{-tL}f)_S$
        \For{$i \in S$}
        \State $g_i \gets \frac{1}{B}\sum_{j = 1}^B f_{H_{ij}}$     
        \EndFor
    \end{algorithmic}
\end{algorithm}

\subsection{Subgradient and stochastic block co-ordinate descent}
We use subgradient descent methods to minimise the penalised loss $F_{t,\la}(\beta;X, y)$ in \eqref{eq:pen_loss}, which we denote henceforth as $F_{t,\la}(\beta)$ for reasons of brevity. The intermediary computation of the subgradients and losses is performed via a heat flow based local dynamics on the network. The precise algorithmic implementation of the latter is given in Algorithms \ref{alg:heatflow} and \ref{alg:heatflow_on_vector}, whereas the complete optimisation procedures (using these as subroutines) are encapsulated in Algorithms \ref{alg:sd} and \ref{alg:block_cd}.

We now describe how the heat flow operator appears in the subgradient computations. Let us compute the subgradient of $F_{t, \lambda}(\beta)$. Set $h = e^{-tL} (\beta \odot \beta)$. Then $\Lambda_t(\beta) = \sum_{j = 1}^p \sqrt{|h_j|}$. Thus
\[
    \frac{\partial \Lambda_t(\beta)}{\partial \beta_{\ell}} = \sum_{j = 1}^p \partial s(h_j) \frac{\partial h_j}{\partial \beta_{\ell}} = \sum_{j = 1}^p \partial s(h_j) 2 (e^{-tL})_{j \ell} \beta_{\ell},
\]
where $s(x) = \sqrt{|x|}$ so that $\partial s(x) = \frac{\sgn(x)}{2 s(x)}$. Set $\zeta = \zeta(\beta) = 2 (\partial s(h_j))_{1 \le j \le p}.$ Then, using the fact that $e^{-tL}$ is symmetric, we can write
\[
    \partial \Lambda_t(\beta) = (e^{-tL} \zeta) \odot \beta.
\]
Therefore
\[
    \partial F_{t, \lambda}(\beta) = \frac{1}{n}X^\top X \beta - \frac{1}{n} X^\top y + \lambda (e^{-tL} \zeta) \odot \beta.
\]
Thus the general subgradient step would be
\[
    \beta^{(m + 1)} = \beta^{(m)} - \alpha^{(m)} \partial F_{t, \lambda}(\beta^{(m)}),
\]
for some learning rare $\alpha^{(m)}$.

We may also use stochastic block co-ordinate descent updates. For this we need to compute $(e^{-tL} \zeta)_S$ for some $S \subset [p]$. This may be efficiently done using an on-demand computation of only the required co-ordinates of $h$, as described in Algorihm~\ref{alg:block_cd}.

\begin{algorithm}[!ht]
    \caption{Subgradient descent}\label{alg:sd}
    \begin{algorithmic}[1]
        \Require Heat flow matrix $H$; initial estimate $\beta^{(0)}$; $\nabla \calL$, gradient of loss function; $\lambda$, penalty parameter; $\epsilon$, error tolerance; $\maxiter$, maximum number of subgradient steps; $\alpha: \N \rightarrow (0, \infty)$, learning rate protocol
        \Ensure $\hat{\beta}$, an approximate local minimum of heat flow penalised objective
        \State $\beta^{(\old)} = \beta^{(0)}$
        \State $\reldiff \gets 2\epsilon$
        \State $\iter \gets 0$
        \While{$\reldiff > \epsilon$ \textbf{and} $\iter \le \maxiter$}
        \State $\iter \gets \iter + 1$
        \State $h \gets \hflow(\beta_{\old}, [p]) \odot \beta_{\old}$     
        \State $\zeta \gets (\ell(h_j))_{1 \le j \le p}$ \Comment{$\ell(x) = \sgn(x)/\sqrt{|x|}$}
        \State $\subgrad \gets \nabla \calL(\beta^{(\old)}) + \lambda \,\, \hflow(\zeta, [p]) \odot \beta^{(\old)}$ 
        \State $\beta^{(\new)} \gets \beta^{(\old)} - \alpha(\iter) \,\, \subgrad$
        \State $\reldiff \gets \frac{\|\beta^{(\new)} - \beta^{(\old)}\|}{\|\beta^{(\old)}\|}$
        \EndWhile
    \end{algorithmic}
\end{algorithm}

\begin{algorithm}[!ht]
    \caption{Stochastic block co-ordinate descent}\label{alg:block_cd}
    \begin{algorithmic}[1]
        \Require Heat flow matrix $H$; initial estimate $\beta^{(0)}$; $\nabla \calL$, gradient of loss function; $\lambda$, penalty parameter; $\epsilon$, error tolerance; $\maxiter$, maximum number of subgradient steps; $\alpha: \N \rightarrow (0, \infty)$, learning rate protocol; $q$, block size for co-ordinate updates
        \Ensure $\hat{\beta}$, an approximate local minimum of heat flow penalised objective
        \State $\beta^{(\old)} = \beta^{(0)}$
        \State $\reldiff \gets 2\epsilon$
        \State $\iter \gets 0$
        \While{$\reldiff > \epsilon$ \textbf{and} $\iter \le \maxiter$}
        \State $\iter \gets \iter + 1$
        \State $S \gets \mathrm{sample}([p], q)$
        \For{$j \in S$}
        \State $e \gets H[j, ]$
        \State $h_e \gets \hflow(\beta^{(\old)}, e) \odot \beta^{(\old)}_e$
        \State $\zeta_e \gets \ell(h_e)$ \Comment{$\ell$ is applied co-ordinatewise}
        \State $g_j \gets \mean(h_e) \odot \beta^{(\old)}_j$ 
        \EndFor
        \State $\subgrad \gets \nabla_S\calL(\beta^{(\old)}) + \lambda \,\, g_S$ 
        \State $\beta^{(\new)}_S \gets \beta^{(\old)}_S - \alpha(\iter) \,\, \subgrad$
        \State $\reldiff \gets \frac{\|\beta^{(\new)} - \beta^{(\old)}\|}{\|\beta^{(\old)}_S\|}$
        \EndWhile
    \end{algorithmic}
\end{algorithm}

\subsection{A final hard-thresholding step}
We apply a hard-thresholding step to the output $\hat{\beta}$ of the optimisation to further enhance its support recovery properties. We apply $K$-means clustering with $K = 2$ on $|\hat{\beta}|$. Let $\cC$ be the cluster whose mean is the closest to $0$. We zero out the values of $\hat{\beta}_{\cC}$. The pseudo-code of this procedure appears in Algorithm~\ref{alg:threshold}.

\begin{algorithm}[!ht]
    \caption{Thresholding}\label{alg:threshold}
    \begin{algorithmic}[1]
        \Require $\beta$, output of subgradient descent and/or stochastic block co-ordinate descent 
        \Ensure $\beta_{\mathrm{thres}}$, a thresholded version of $\beta$
        \State $(C_1, m_1), (C_2, m_2) \gets \textsc{kmeans}(|\beta|, 2)$
        \State $j_* \gets \arg \max_{j \in \{1, 2\}} m_i $
        \State $\beta_{\mathrm{thres}} \gets \beta \odot \bone_{C_{j_*}}$
    \end{algorithmic}
\end{algorithm}

\subsection{Learning the graph Laplacian from data}
If the Laplacian $L$ (equivalently, the graph $G$) is not a priori available, we may estimate it using the covariance structure of the explanatory variables as follows (the pseudo-code appears in Algorithm~\ref{alg:graph_estimation}). Let $\hat{\Sigma}$ be some estimate of the true covariance matrix $\Sigma$. Let $\hat{R}$ denote the estimate of the true correlation matrix, obtained by rescaling $\hat{\Sigma}$. We construct an adjacency matrix by thresholding $\hat{R}$:
\[
    A_{ij} = \bone_{\{|\hat{R}_{ij}| \ge \tau(\hat{R})\}}, 
\]
where $\tau(\hat{R})$ is an appropriate quantile of the $|\hat{R}_{ij}|$'s. (E.g., in our simulation experiments, we use the $0.75$-th quantile.) We then take the unnormalised Laplacian corresponding to $A$ as an estimate of $L$.

As for estimating $\Sigma$, we may use the sample covariance matrix in low to moderate dimensions, and appropriately thresholded versions of it (e.g., \cite{bickel2008covariance, bickel2008regularized}) or other shrinkage estimators in high dimensions (e.g., \cite{friedman2008sparse, chen2010shrinkage}).

\begin{algorithm}[!ht]
    \caption{Graph estimation}\label{alg:graph_estimation}
    \begin{algorithmic}[1]
        \Require $\hat{R}$, some estimate of the population correlation matrix of the covariates; $\alpha$, quantile at which to threshold
        \Ensure $A$, adjacency matrix of a graph constructed from $\hat{R}$
        \State $\theta \gets \mathrm{quantile}(|\hat{R}|, \alpha)$
        \State $A \gets \bzero_{p \times p}$
        \For{$i = 1, \ldots, (p - 1)$}
        \For{$j = (i + 1), \ldots, p$}
        \State $A_{ij} \gets \ind_{\{|\hat{R}_{ij}| > \theta\}}$
        \State $A_{ji} \gets A_{ij}$
        \EndFor
        \EndFor
    \end{algorithmic}
\end{algorithm}

\subsection{Optimisation via local network dynamics} 

A cornerstone of our algorithms is that we do not require any direct knowledge of the group structure. Clustering algorithms for estimating the group structure typically require the knowledge of the number of groups. Thus in order to use group Lasso with an estimated group structure, one \emph{needs to know the number of groups}. Our algorithms, on the other hand, do not have such a requirement.

Further, we do not even require access to the full network. Rather, at each step of the optimisation algorithm, we find an approximation to the subgradient of the penalty (at the current value of the parameter $\b$) by using an Monte Carlo based approach. Additional economy in computational resources is accorded by the fact that we are able to use the same heat flow throughout the optimization procedure, the initial and terminal nodes of which generated and stored for successive calls. The effectiveness of this economizing protocol is demonstrated in simulation experiments as well as applications to real world data.

In particular, running a continuous time random walk only entails exploring a local neighbourhood of the current state, which relieves us of the necessity to work with the entire graph at one go. The latter would be necessary, e.g., in an approach where we wanted to do an initial spectral clustering  in order to arrive at a detailed understanding of the group structure  \citep{buhlmann2013correlated}. This may be prohibitively expensive in the setting of real-world massive networks, such as the large-scale social networks or the world-wide web.

In fact, we only  require \textit{oracle access} to a black box that returns the terminal state of the heat flow starting from a prescribed initial node. Such limited and local  access to the network data can have significant implications with regard to considerations of privacy and security, which play an increasingly important role in modern statistical research. 

\subsection{The duration and  step count for the heat flow dynamics}
The duration of the heat flow is determined by the goal to make the difference between our penalty and the classical group lasso penalty small,  a bound on which is accorded by Theorem \ref{thm:prediction-consistency}. It follows therefrom  that,  under reasonable conditions,  it suffices to have $\tflow \gg \frac{1}{\la_g}\max \{ \log n, \log p\}$.  In the most important setting of the clusters being fully disconnected from each other but densely connected within each other,  the ground state $\la_g$ is the minimum of the ground states for the individual clusters. For the components being generic densely connected graphs of size $\Theta(p)$, we typically have $\la_g=\Theta(p)$,  whence $\tflow \gg \frac{1}{p}\max \{ \log n, \log p\}$; we refer to Section 3 in the appendix for details. We note that although the optimal choice of $t_{\mathrm{flow}}$ depends on $\lambda_g$, in practice we can perform \emph{cross-validation} to choose the optimal value of $t_{\mathrm{flow}}$. It is precisely for this reason that our algorithms do not need the knowledge of the number of groups. 

The number of steps of the heat flow dynamics is,  roughly speaking,  the heat flow time $\tflow$ times the average number of steps per unit time.  If the current state of the dynamics is a node $v \in G$ with degree $\mathrm{deg}(v)$,  then the time until the next step in the random walk is an $\mathrm{Exponential}(\mathrm{deg}(v))$) random variable.  Thus,  step time distribution is stochastically dominated by an $\mathrm{Exponential}(d_{\max})$ random variable,  where $d_{\max}$ is the maximum degree of $G$. Thus,  on average,  the total step count for the heat flow dynamics is given by $\Ns=O(d_{\max} \cdot \tflow)=O(d_{\max} \cdot \frac{1}{\la_g} \cdot \max \{ \log n, \log p\})$.  For typical strongly intra-connected components as above, we have $d_{\max}=\Theta(p)$ (for details we refer to Section 3 in the appendix), whence we have $\Ns = O(\max \{ \log n, \log p\})$.  

The logarithmic dependence on the problem dimensions as indicated above entails a light computational load for the heat flow based dynamic algorithm. 

\subsection{A spectral perspective and the role of non-convexity}

We note in passing that the spectral data of a heat flow operator (equivalently,  an appropriate random walk transition matrix) is known to be useful for learning tasks,  especially in the context of diffusion mapping for dimensionality reduction problems \citep{coifman2005geometric, coifman2006diffusion}.  However,  the random walk in diffusion mapping takes place in a different space -- namely,  on the space of the actual data points (therefore,  having $n$ nodes); whereas  in our setting,  the random walk takes place on the co-ordinate indices of the data points (thereby entailing $p$ nodes).  Nonetheless,   it would be of  interest to explore the possible interaction of these two diffusion-based approaches,  in particular regarding the possibility of incorporating a dimension reduction step in our paradigm to achieve further economy of computational resources.

It may be noted that the map $\Psi$ that is embedded in our penalty is mildly non-convex, with an algebraically explicit square-root structure. This may be compared to the setting of classical group lasso when the groups are a priori not known, where  non-convexity enters via its role in estimating the groups. This step involves solving a  clustering problem, which is generally non-convex.  In fact, the non-convex structure of the group recovery problem is known for its notoriety as a challenging optimization issue, and its rather intractable combinatorial nature makes it arguably a more complicated endeavour than the simple, algebraically explicit non-convexity posed by the map $\Psi$ in our approach. Exploiting this simple algebraic structure, perhaps by carrying out the  optimization in a different co-ordinate system, would be an interesting direction for future research.

\section{Theoretical guarantees}\label{sec:main}
We first set some notations in order to lay out our main theoretical results. For definiteness, we focus on the setting of  regression with a group structure  on the parameters in our theoretical analysis. However, we note in passing that similar analysis would apply to a wide range of applications with our method, including logistic regression, generalised linear models and other use cases. 

Consider the situation where the graph $G$ has exactly $k$ connected components. Then it is a well-known fact that the spectrum $0 = \la_0 \le \la_1  \le \cdots  \le \la_p$ of the Laplacian matrix $L$ of $G$ then has exactly $k$ many zero eigenvalues. Let $\lambda_g := \min_{i > k} \lambda_i$ denote the \textit{spectral gap} of $L$. Let
\[
    A(\beta) = \{ i \mid 1 \le i \le k, \|\beta_{\cC_i}\|_2 \ne 0\},
\]
and
\[
    I(\beta) = \cup_{i \in A(\beta)} \cC_i.
\]

Note that $\Lambda_t(\beta)$ is a non-convex penalty. As such, the objective function $F_{t, \lambda}$ can have multiple local optima. However, we will show that in an appropriate neighbourhood of the true parameter value $\beta^*$, the minimum of $F_{t, \lambda}(\beta)$ approximately minimises the group Lasso penalty for sufficiently large $t$.

As we shall see below, the presence of an additional \textit{restricted eigenvalue property} leads to improved guarantees on the accuracy of our procedure, so we state this property below.

\textit{Property RE($s$)}:  We say that the restricted eigenvalue property RE($s$) holds for $X$ with parameter $\kappa = \kappa(s)$ if
\begin{align*}
    \min\bigg\{\frac{\|X\Delta\|}{\sqrt{n\|\Delta_A\|}} \,& :\, |A|    \le s,  \quad \Delta \in \R^p \setminus \{0\}, \\ & \sum_{j \notin A} \sqrt{|\calC_j|} \|\Delta^j\| \le 3 \sum_{j \in A} \sqrt{|\calC_j|} \|\Delta^j\|\bigg\} \ge \kappa.
\end{align*}

We denote $B(\beta^*,\eps)$ to be the $L^2$ Euclidean ball in $\R^p$ with center $\beta^*$ and radius $\eps$, and set 
\[
    \Lambda(X; \eta) = \max_{1 \le j \le K} \sqrt{\|\frac{1}{n}X_j^\top X_j\|_{\op}}\bigg(1 +  \sqrt{\frac{4\log \eta^{-1}}{|\calC_j|}}\bigg).
\]
We may then state:

\begin{theorem}\label{thm:prediction-consistency}
    For $0 < \epsilon \le \epsilon_0$, let $\hat{\beta}_{t, \lambda} = \hat{\beta}_{t, \lambda}^{(\epsilon)}$ be the minimiser of $F_{t, \lambda}(\beta)$ in $B(\beta^*, \epsilon)$. 
    Let $\frac{\sigma}{\sqrt{n}} \Lambda(X; \eta) \le \la \le \l(8 |\cmax| \lambda \|\beta^*\|_{2, 1}\r)^{-1} \eps$.
    Then  with probability at least $1 - 2K\eta$, we have
    \begin{equation} \label{eq:PC_slow}
        \frac{1}{n}\|X(\hat{\beta}_{t, \lambda} - \beta^*)\|_2^2 = O(\|\beta^*\|_{2, 1} \lambda |\cmax| + p^{3/2} e^{-t\lambda_g/2}).
    \end{equation}    
    If we further assume $\RE(s)$ holds for $X$ with parameter $\kappa$, then
    \begin{equation} \label{eq:PC_fast}
        \frac{1}{n}\|X(\hat{\beta}_{t, \lambda} - \beta^*)\|_2^2 = O\bigg(\frac{s\lambda^2 |\cmax|}{\kappa^2} + p^{3/2} e^{-t\lambda_g/2}\bigg).
    \end{equation}
\end{theorem} 

We further demonstrate that the RE($s$) property above for a random design matrix $X$ can be related, in the crucial setting of Gaussian covariates, to a similar property for the \textit{deterministic} covariance matrix $\Sigma$ of the rows of $X$. The latter can often be much easier to verify; for instance, it may be easily seen to hold as soon as $\Sigma$ is well-conditioned. 

\begin{lemma}\label{lem:res-eig}
    Suppose that $\Sigma$ satisfies $\RE(s)$ with parameter $\kappa_{\Sigma}(s) > 0$. Let the rows of $X$ be i.i.d. $N(0, \Sigma)$. Assume that
    \begin{equation}
        n \ge (36 \cdot 8)^2 \cdot \frac{(\rho(\Sigma))^2 s |\cmax| \log p}{(\kappa_{\Sigma}(s))^2}.
    \end{equation}
    Then $\frac{X^\top X}{n}$ satisfies $\RE(s)$ with parameter $\kappa_{\Sigma}(s)/8$ with probability at least $1 - c'\exp(-cn)$ for some constants $c, c' > 0$.
\end{lemma}

\section{Random designs with a latent network geometry}

\subsection{Gaussian Graphical Models and Gaussian Free Fields}
\textit{Gaussian Free Fields} (abbrv.  GFF) have emerged as  important models of  correlated Gaussian fields,  that are naturally commensurate with the geometry of their ambient space.  In the case of graphs, the ambient geometry is spawned by the graph Laplacian.  These Gaussian processes also have important applications in physics, where they are of interest in the context of Euclidean quantum field theories \citep{friedli2017statistical}.

GFFs are in fact Gaussian Graphical Models (abbrv. GGM),  where the precision matrix of the Gaussian random field aligns with the Laplacian  of the graph,  thereby leading to a rich interaction between the statistical properties of the GGM and the Laplacian geometry of the underlying graph \citep{zhu2003semi,zhu2003combining,ma2013sigma,kelner2019learning,rasmussen2003gaussian}.  GGMs have emerged as   popular tools to model dependency relationships in data via a latent graph structure,  the choice of Gaussian randomness being often motivated by the fact that a Gaussian distribution maximises entropy within the constraints of a given covariance structure.  Applications of GGMs are ubiquitous, with use cases in diverse domains such as structural inference in biological networks, causal inference problems, speech recognition, and so on \citep{whittaker2009graphical,lauritzen1996graphical,edwards2012introduction,uhler2019gaussian}.  Since our approach  exploits in an essential manner the Laplacian geometry of the graph, the GFF is a natural GGM to examine within its ambit. 

For an in-depth introduction to the technical aspects of GFFs, we refer the reader to the excellent surveys \cite{sheffield2007gaussian} and \cite[Chap. 1]{berestycki2015introduction}; herein we will content ourselves with a brief description of its relevant features.   Broadly speaking,  a GFF is essentially a natural generalization of Brownian motion to general spaces,  with time replaced by,  e.g.,  the nodes of a network. 
We define the \textit{massive} GFF $\X_\t=(\X_\t(v))_{v \in V}$ on a graph $G=(V,E)$ with mass parameter $\t>0$ to be a mean-zero Gaussian field indexed by $V$ characterised by its \textit{precision matrix} (i.e., the inverse of the covariance matrix) given by $(L + \t I_{|V|})$, where $L$ is the (unnormalised) graph Laplacian on $G$ and $I_{|V|}$ is the $|V| \times |V|$ identity matrix (c.f. \cite{berestycki2015introduction}). The 
covariance matrix of $\X_\t$ is therefore given by $\S=(L + \t I_{|V|})^{-1}$. Note that $L$ itself is singular due to the all ones vector $\mathbbm{1}_V$ being in its kernel;   therefore we are aided by the strong convexity accorded by the mass parameter $\t>0$. 

\subsection{Stochastic Block Models and random designs}
In this section, we consider a very different model of graph clustering that is motivated by Stochastic Block Models (abbrv.,  SBM) that have attracted intense focus in statistics and machine learning applications in recent years  \citep{abbe2017community,goldenberg2010survey,holland1983stochastic,karrer2011stochastic}.  SBMs are  underpinned by a block matrix structure, indexed by the vertices of a graph,  with entries in $[0,1]$ that are constant on the blocks. These entries are the connection probabilities between the respective vertices.  In our setting,  it would be natural to consider the same block matrix pattern to generate the graph and the covariance structure.

For definiteness,  we divide the vertex set $V$ in $k$ groups $\{\calC_i\}_{i=1}^k$. We consider the $|V| \times |V|$  matrix $\cP$ whose rows and columns are indexed by the nodes of the graph $G$. For  $1 \le i,j \le k$,  let $\cP^{i,j}$ denote the submatrix indexed by the vertices in the groups $\calC_i$ (along the rows) and $\calC_j$ along the columns. For a vector $v \in \R^{|V|}$, we denote by $D(v)$ the diagonal matrix whose diagonal equals $v$. For brevity, we will denote by $\ind_i$ the indicator vector (in $\R^{|V|}$) of the nodes in $\calC_i$.  Let $a>b$ be numbers in $[0,1]$, such that $\cP^{i,i} = a (\ind_i \ind_i^\top -  D(\ind_i))$, and for $1 \le i \ne j \le k$, we have $\cP^{i,j} = b \ind_i \ind_j^\top$. Typically, $k$ is small compared $|V|$ and $a-b$ is taken to be large enough.

The classical SBM is a random graph that is sampled with independent edges according to probabilities given by the matrix $\cP$.  In our setting, the graph $G$ would simply be a realisation of this SBM, with the parameter $b$ being small so as to ensure the graph to be sparsely connected across blocks.  An alternative paradigm of graph generation from the matrix $\cP$ would be to consider a deterministic,  albeit weighted,  graph $G$,  whose edge weight between the vertices $u,v \in V$ equals $\cP_{uv}$.  The correlation structure of the covariates for either model is given  by the covariance matrix $\S=I_{|V|} + \cP$. In practice, if a graph is constructed from such a block-structured sample covariance matrix via thresholding, then the resulting graph would be well-modelled by the above SBM.

\subsection{Prediction guarantees and sample complexity bounds}  \label{sec:recovery}
In this section, we provide ballpark estimates on the quantities of interest in our random design models, that would be  enough to guarantee the effectiveness of our approach. We record here the main conclusions of our analysis, postponing the details to the appendix. 

For definiteness, we fix a polynomial decay of probability with which our recovery guarantees are to hold, which implies that the quantity $\eta=O(n^{-\a})$ in Theorem \ref{thm:prediction-consistency} for a fixed $\a>0$. Under this error tolerance, the appropriate choice of $\la$ may be shown to be $\la \gtrsim  \s_{\max}(\S) \sqrt{\frac{\log n}{n}}$, which in turn leads to  a prediction guarantee of 
\begin{align*} 
    \frac{1}{n}\|X(\hat{\beta}_{t, \lambda} - \beta^*)\|_2^2 \numberthis  \label{eq:PG} = O_P\l( \frac{\log n}{n}  \cdot \frac{ \s_{\max}(\S)}{\s_{\min}(\S)^2} \cdot  s |\calC_{\max}| + p^{3/2} e^{-t\lambda_g/2}  \r). 
\end{align*} 

Our goal here is to understand the order of the flow time $\tflow$ and the step count $\Ns$ (in terms of the other parameters of the problem) required to achieve a desired accuracy.  We will  bifurcate our analysis into two related parts.

\subsubsection{Bounds on \texorpdfstring{$\tflow$}{} for given \texorpdfstring{$n$}{}}
Given the data size $n$, we investigate the order of $t=\tflow$ at which the approximation error due to our heat flow based approach (roughly, the second term in \eqref{eq:PG}) becomes comparable to  the contribution to the prediction error bound for the classical group lasso methods that assume complete knowledge of the group structure (roughly, the first term in \eqref{eq:PG}). It may be shown that under very general circumstances, we have
\begin{equation} \label{eq:time_10}
    \tflow ~\gtrsim \frac{1}{\la_g} \log p + \frac{1}{\la_g} \log \l(  \frac{ n }{\log n} \cdot \frac{ \s_{\min}(\S)^2} { \s_{\max}(\S)} \cdot \frac{1} {s |\calC_{\max}|  }\r), 
\end{equation}
whereas for most models of interest, including the GFF and SBM based models in our purview, we may deduce the much simpler prescription $\tflow ~\gtrsim \frac{1}{\la_g} \max \{ \log p , \log n \}$. The details are provided in Section 3 of the appendix.

\subsubsection{Bounds on \texorpdfstring{$\tflow, n$}{} for target prediction guarantee \texorpdfstring{$\eps$}{}}

We fix a  threshold $\eps$, and make explicit prescriptions for the order of $n$ and $\tflow$ that will allow us to obtain a prediction error of order $O(\eps)$. To this end, we posit that the two terms on the right hand side of  \eqref{eq:PG} are separately $O(\eps^2)$. For $\tflow$, this entails that $p^{3/2}e^{-\la_g \tflow} \lesssim \eps^2$, which translates into $\tflow \gtrsim \frac{1}{\la_g}\l( \log p + \log \frac{1}{\eps} \r)$.

For a prescription for $n$ given a target prediction error $\eps$, we have the bound
\begin{align*}
 \frac{n}{\log n}  \numberthis \label{eq:n_eps}  \gtrsim \max \l\{  \l( \frac{1}{\eps^2} \cdot s |\calC_{\max}| \cdot \frac{ \s_{\max}(\S)  }{ \s_{\min}(\S)^2}\r), \l( s |\calC_{\max}| \cdot \frac{\rho(\Sigma)^2  \log p}{\s_{\min}(\S)^2} \r) \r\}. 
\end{align*}
The details of the analysis are available in Section 3 of the appendix.

\subsection{Thresholds for typical clustered networks}
For typical clustered networks (for a concrete probabilistic model, see Section 3 in the appendix), we can further simplify above prescriptions on the heat flow time $\tflow$  to the thresholds 
\begin{equation} \label{eq:eq:time_well_clustered}
    \tflow ~\gtrsim \frac{1}{p} \cdot  \max \{ \log p , \log n \}; \quad \tflow ~\gtrsim \frac{1}{p} \cdot \max \{ \log p , \log \frac{1}{\eps} \},
\end{equation}
resp. for the settings where data size $n$ is given and where the target prediction error $\eps$ is given. These translate into the step count bounds
\begin{equation} \label{eq:eq:step_well_clustered}
    \Ns = O\l( \max \{ \log p , \log n \} \r); \; \Ns = O\l(\max \{ \log p , \log \frac{1}{\eps} \} \r).
\end{equation}
We refer the reader to Sec. 3 in the appendix for details. It may be noted that the flow time $\tflow$ and the step count $\Ns$ do not depend on the statistical properties of the covariates $X_i$, but depend only on the geometric properties of the network $G$.

\subsubsection{Explicit guarantees for GFF and block model designs}

In the setting of GFF and SBM based random designs, we may obtain further explicit guarantees on the sample complexity $n$, via application of techniques from spectral graph theory; the detailed analysis is provided in Section 3 of the appendix. 

For both SBM based  and  GFF based random designs on  well clustered networks, we have the following prescriptions that suffice with high probability:
\begin{equation} \label{eq:sample_cplx_GFF}
    \l( \frac{n}{\log n} \r)_{\text{GFF}}  \gtrsim \max \l\{ \frac{1}{\eps^2} \cdot \|\b^*\|_0 \cdot  p,~~ \|\b^*\|_0 \cdot \log p \r\} 
\end{equation}
and
\begin{equation} \label{eq:sample_cplx_SBM}
    \l( \frac{n}{\log n} \r)_{\text{SBM}}  \gtrsim \max \l\{ \frac{1}{\eps^2} \cdot \|\b^*\|_0,~~ \|\b^*\|_0 \cdot \log p \r\}. 
\end{equation}

We observe that  the sample complexity bound \eqref{eq:sample_cplx_SBM}  for an SBM based design matches, upto logarithmic factors, the analogous bound for classical sparse reconstruction problems. For GFF based random designs,  it may be noted that the $\frac{1}{\eps^2} \cdot \|\b^*\|_0 \cdot  p$ term in \eqref{eq:sample_cplx_GFF}  comes from the first term in \eqref{eq:n_eps}, which, roughly speaking, reflects  the error incurred by classical group lasso. Thus, the linear dependence of the sample complexity on $p$ appears to be a fundamental characteristic of the problem for GFF based random designs, and is inherent to both classical group lasso and the present heat flow based methods. Empirical investigations also appear to corroborate this effect; for details we refer the reader to Section 4 in the appendix. Although  our method is still effective in a high dimensional setup (up to logarithmic factors),  devising  methodologies with greater efficiency for GFF based random designs, as well as theoretical investigations of information theoretic lower bounds in this setting, would be an interesting direction for future research.

\section{Experiments}\label{sec:exp}
\subsection{Simulations}
We take $n = 200, p = 100$ and $k = 4$ groups of relative sizes $(p_1, p_2, p_3, p_4)^\top / p = (0.16, 0.24, 0.40, 0.20)^\top$. We denote these groups by $\cC_i, i = 1, \ldots, k$.

We consider two models for $X$:
\begin{enumerate}
    \item \textbf{Gaussian with block diagonal covariance matrix.}
        The covariates $X \sim \mathcal{N}(0, \Sigma)$, with
        \begin{equation}\label{eq:block_diagonal_sigma}
            \Sigma = \begin{pmatrix}
                \Sigma_{p_1}(\rho_1) & \bzero & \bzero & \bzero \\
                \bzero & \Sigma_{p_2}(\rho_2) & \bzero & \bzero \\
                \bzero & \bzero & \Sigma_{p_3}(\rho_3) & \bzero \\
                \bzero & \bzero & \bzero & \Sigma_{p_4}(\rho_4)
            \end{pmatrix},
        \end{equation}
        where $\Sigma_d(\rho) = (1 - \rho) I_d + \rho \ind_d \ind_d^\top$ is the equi-correlation matrix of order $d$.

    \item \textbf{Gaussian free field.}
        We take a graph $G$ generated from a stochastic blockmodel on $p$ vertices with the groups $\cC_i, i = 1, \ldots, k$, and a connection probability of $a = 0.5$ within groups and $b = 0.01$ between groups. We then take a massive GFF on the graph $G$, with mass parameter set to the $(k + 1)$-th smallest eigenvalue of the unnormalised Laplacian corresponding to $G$.
\end{enumerate}

The true parameter $\beta$ is generated in the following way
\[
    \beta_i \begin{cases}
        \sim \mathrm{Uniform}(0.5, 0.7) & \text{ for } i \in \cC_1, \\
        = 0 & \text{ for } i \in \cC_2, \\
        \sim \mathrm{Uniform}(-0.7, -0.5) & \text{ for } i \in \cC_3, \\
        = 0 & \text{ for } i \in \cC_4.
    \end{cases}
\]

Finally, we generate the response from the linear model
\[
    Y = X\beta + \eps,
\]
where $\eps$ is a noise vector, independent of $X$, with isotropic covariance matrix $\sigma^2 I_n$.

We compare the subgradient and the stochastic block co-ordinate descent versions of our procedure (referred to as ``Heat flow (SD)'' and ``Heat flow (CD)'', respectively) against group lasso with group structure learned from spectral clustering \citep{von2007tutorial} on $\hat{L}$, an estimate of the Laplacian matrix obtained using Algorithm~\ref{alg:graph_estimation} with oracle knowledge of $k = 4$. For each algorithm, the tuning parameter $\lambda$ is chosen by cross-validation. In Tables~\ref{table:sim-1} and \ref{table:sim-2}, we report the prediction error, estimation error, and two measures of support recovery, namely, sensitivity and specificity:
\begin{align*}
    \text{Sensitivity} &= \frac{\#\{i : \hat{\beta}_i \ne 0, \beta^*_i \ne 0\}}{\#\{i : \beta^*_i \ne 0\}}, \\
    \text{Specificity} &= \frac{\#\{i : \hat{\beta}_i = 0, \beta^*_i = 0\}}{\#\{i : \beta^*_i = 0\}}.
\end{align*}
The estimated $\beta$ for different methods in the experiment with block diagonal covariate structure is shown in Figure~\ref{fig:sim-1}.

From both the experiments we see that our proposed method based on the heat flow penalty has comparable performance to group lasso in terms of prediction/estimation error, without requiring an explicit knowledge of the group structure, or even of the number of groups. In terms of support recovery, the proposed method appears to be much superior.

\begin{figure}[!ht]
    \centering
    \begin{tabular}{c}
        \includegraphics[scale = 0.50]{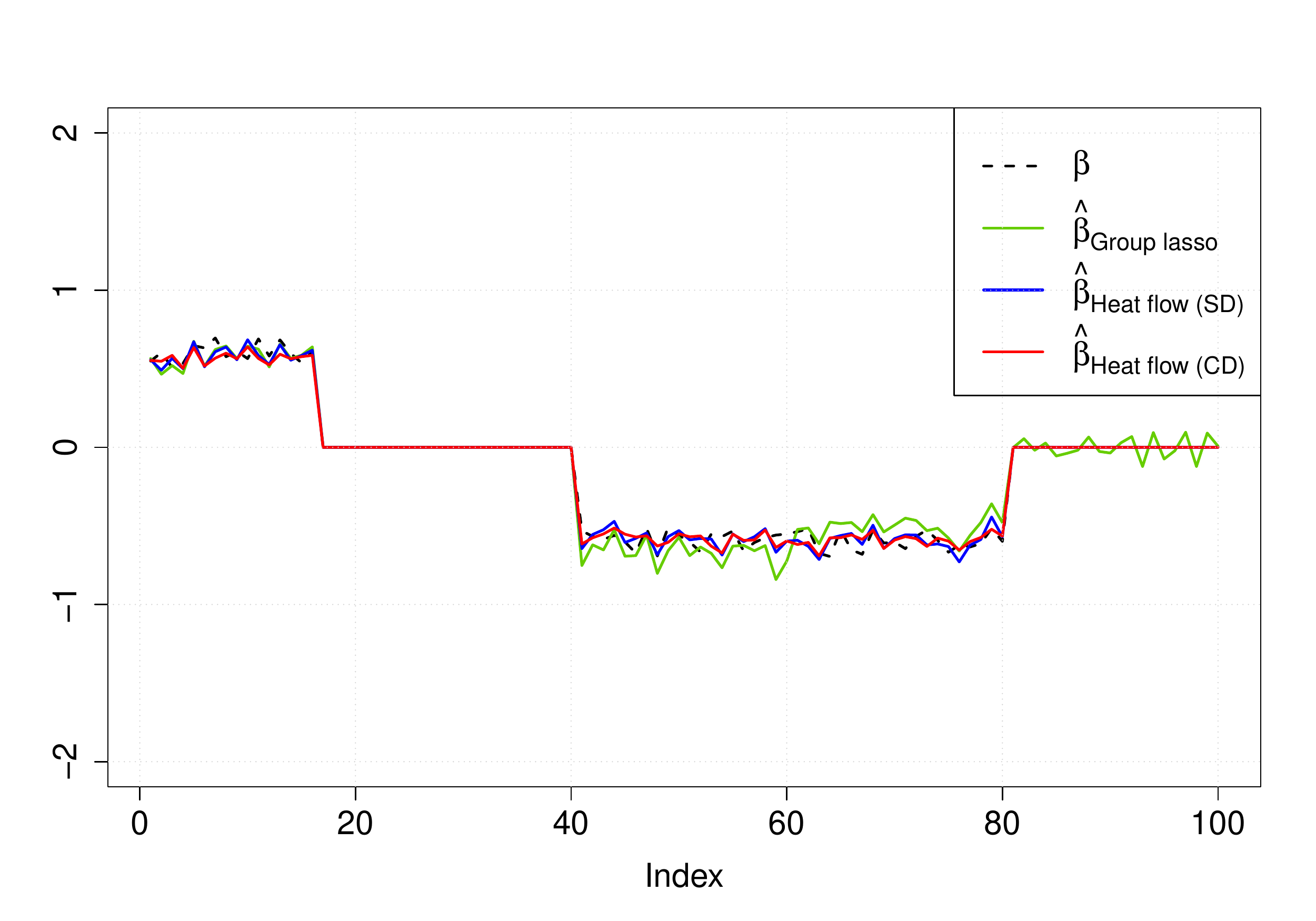}
    \end{tabular}
    \caption{Estimated $\beta$ for different methods in the simulation experiment with block diagonal covariance structure.}
    \label{fig:sim-1}
\end{figure}

\begin{table}[!ht]
    \centering
    \begin{tabular}{c|c|c|c}
        \toprule
                         & Group lasso & Heat flow (SD) & Heat flow (CD) \\
                     \midrule
        Prediction error & 0.03        & 0.02           & 0.03           \\
        Estimation error & 0.84        & 0.50           & 0.49           \\
        Sensitivity      & 1.00        & 1.00           & 1.00           \\
        Specificity      & 0.55        & 1.00           & 1.00           \\
        \bottomrule
    \end{tabular}
    \caption{Comparison when $X \sim \mathcal{N}(0, \Sigma)$ with a block diagonal $\Sigma$ as in \eqref{eq:block_diagonal_sigma}. We take $(\rho_1, \rho_2, \rho_3, \rho_4)^\top = (0.6, 0.9, 0.7, 0.4)^\top$.}
    \label{table:sim-1}
\end{table}

\begin{table}[!ht]
    \centering
    \begin{tabular}{c|c|c|c}
        \toprule
                         & Group lasso & Heat flow (SD) & Heat flow (CD) \\
                     \midrule
        Prediction error & 0.12        & 0.14           & 0.14           \\
        Estimation error & 3.67        & 4.10           & 3.95           \\
        Sensitivity      & 0.45        & 0.77           & 0.39           \\
        Specificity      & 1.00        & 0.98           & 0.98           \\
        \bottomrule
    \end{tabular}
    \caption{Comparison when $X$ is drawn from a Gaussian free field.}
    \label{table:sim-2}
\end{table}

\subsection{Real data}
In this section, we compare our methods against group lasso in terms of test-set performance in four real-world data sets. For group lasso, the group structure is learned by applying spectral clustering on $\hat{L}$, the estimated Laplacian matrix obtained by applying Algorithm~\ref{alg:graph_estimation}, with $k$ set to be the number of eigenvalues of $\hat{L}$ less than $0.01$. For each data set, we use an 80:20 split into training and test sets. The test-set errors are reported in Table~\ref{table:real_data}. We observe comparable performance in all four data sets.

\subsubsection{Application to email spam data}
We consider the well-known spambase data set\footnote{https://archive.ics.uci.edu/ml/datasets/spambase} containing 4601 emails classified as spam/non-spam. There are 57 explanatory variables. We fit a logistic regression model with group lasso and heat flow penalties. We report the test-set misclassification error in the second column of Table~\ref{table:real_data}.

\subsubsection{Application to gene-expression data}
We consider gene expression data from the microarray experiments of mammalian eye tissue samples of \cite{scheetz2006regulation}. The response variable is the expression level of the TRIM32 gene. There are 200 predictor variables corresponding to different gene probes. The sample size is 120. We report the test-set mean-squared error (MSE) in the third column of Table~\ref{table:real_data}.

\subsubsection{Application to climatological data}
From the NCEP/NCAR reanalysis data set, we took the monthly average temperature of the Delhi-NCR region as the response variable. We took monthly average temperature, pressure, precipitation, wind-speed, etc. of $2.5^\circ \times 2.5^\circ$ blocks on the Bay of Bengal and the Arabian Sea as covariates. In total there were 101 such blocks, giving us $p = 606$ explanatory variables. We have these measurements for $n = 886$ months, starting January, 1947 till October 2021. We first removed seasonal variations and fitted a linear trend afterwards as preprocessing steps (as described in \cite{chatterjee2012sparse}). We report the test-set MSE in the fourth column of Table~\ref{table:real_data}.

\subsubsection{Application to stock-market data}
We have data on daily highs of the NIFTY 50 index from the National Stock Exchange (NSE) of India for 49 companies for $n = 2598$ days staring from November 4, 2010 till April 30, 2021. We use the mean index of $9$ companies in the financial sector as our response variable, and use the indices of the rest of the $p = 40$ companies as covariates. We report the test-set MSE in the fifth column of Table~\ref{table:real_data}.

\begin{table}[!htbp]
    \centering
    \begin{tabular}{c|c|c|c|c}
        \toprule
        Method         & spam   & gene   & climate & stock  \\
        \midrule
        Group lasso    & 0.18 & 0.15 & 0.06  & 0.02 \\
        Heat flow (SD) & 0.11 & 0.15 & 0.07  & 0.02 \\
        Heat flow (CD) & 0.11 & 0.15 & 0.07  & 0.03 \\
        \bottomrule
    \end{tabular}
    \caption{Test-set error on real data.}
    \label{table:real_data}
\end{table}

\section{Conclusion}
In this work, we contribute an approach to learning under a group structure on explanatory variables that does not require prior information on the group identities. Our paradigm is motivated by the Laplacian geometry of an underlying network with a commensurate community structure, and proceeds by directly incorporating this into the penalty. In a more general setup, when an underlying graph may not be explicit in the problem description, we demonstrate a procedure to construct such a network based on the available data. Notably, we dispense with the elaborate pre-processing step involving clustering of the variables, spectral or otherwise, which can be computationally resource-intensive. Our paradigm is underpinned by rigorous theorems that guarantee effective performance and provide bounds on its sample complexity. In particular, we demonstrate that in a very wide range of settings, we need to run the heat flow dynamics for a time that is only logarithmic in the problem dimensions. We investigate in detail the interplay of our approach with key statistical physics paradigms such as the GFF and the SBM. We validate our approach by successful application to real-world data from diverse fields including computer science, genetics, climatology and economics.

Our approach opens the avenue to applications of similar dynamical techniques to classical statistical and data analytical problems, that are normally defined as static problems. The inherently local nature of the heat flow and related diffusion dynamics enables us to resolve the relevant constrained optimization problems while being oblivious to the global geometry of the graph (such as a complete understanding of the clustering structure of the variables). In addition to economies of computational resource, such locality is of significance in the context of questions of privacy in data analytical methodologies, a problem that is gaining increasing salience in today's hyper-networked world. On a related note, it would be of interest to enhance our approach to obtain similarly local algorithms that address additional structural features of the explanatory variables, such as smoothness or intra-group sparsity. Yet another intriguing direction would be to explore the interface of our approach and diffusion-mapping based techniques that have been effective for dimension reduction problems, and exploit their interplay to achieve further economy of scale and computational resources. In general, the interplay between the geometric structure provided by the Laplacian, the stochastic structure accorded by models such as the GFF and SBM and the inherent clustering structure of real world datasets raises the possibility of a rich mathematical theory and a suite of associated techniques to evolve.

\section*{Acknowledgements}
S.G. is supported in part by the MOE grants R-146-000-250-133 and R-146-000-312-114. S.S.M. is supported by an INSPIRE Faculty Fellowship from the Department of Science and Technology, Government of India. The authors thank Snigdhansu Chatterjee for pointers to the NCEP/NCAR reanalysis data set.

\bibliographystyle{apalike}
\bibliography{glasso}

\appendix
\renewcommand{\theequation}{S.\arabic{equation}}
\setcounter{equation}{0}

\section{Generalities}

\subsection{The generator of the heat flow}

In this section we demonstrate that the generator of the heat flow in Algorithm 1 is indeed the graph Laplacian. To this end, we denote $X_t$ to be the location of this continuous time Markov Chain at time $t$. Let $f$ be a test function on the graph $G$. 

Let $v \in G$ be a vertex with degree $\deg(v)$. For $\del>0$ small, we proceed to compute $\E[f(X_{t+\del}) | X_t = v]$. The probability of there being multiple jumps of the Markov Chain in time $\del$ is $O(\del^2)$, and therefore, to understand the above expectation to the first order in $\del$, we focus on the situation where there is at most one jump in the time interval $(t,t+\del)$.

The next jump in the Markov Chain occurs when the exponential clock along any of the edges of $G$ incident on $v$ rings. Since these clocks are i.i.d. with parameter 1 each, we  the timing of the next jump is the minimum of $\deg(v)$ many i.i.d. Exponential (1) random variables. The latter random variable is easily verified to be an Exponential ($\deg(v)$) random variable, whose mean is $1/\deg(v)$. 
We have, $\P[\text{No jump in } (t,t+\del)] = \exp(-\del \deg(v))$.  If there is a jump in the time interval $(t,t+\del)$, the Markov Chain moves to a neighbouring vertex of $v$ chosen uniformly at random, each with probability $1/\deg(v)$. 

Therefore, we may write 
\begin{equation}
    \E[f(X_{t+\del}) | X_t = v] = \exp(-\del \deg(v)) \cdot f(v) + (1 - \exp(-\del \deg(v))) \cdot \frac{1}{\deg(v)} \cdot \l( \sum_{u \sim v} f(u) \r).
\end{equation}

This implies that 
\begin{align*} 
    \E[f(X_{t+\del}) &| X_t = v] - f(v) \\
    = & \l(1 - \exp(-\del \deg(v))\r) \cdot \l[\frac{1}{\deg(v)} \cdot \l( \sum_{u \sim v} f(u) \r)   - f(v)\r] \\
    = & \l(1 - \exp(-\del \deg(v))\r) \cdot \frac{1}{\deg(v)} \cdot \l[ \sum_{u \sim v} (f(u) - f(v)) \r]. 
\end{align*}

To the end of computing the generator $\cG$ of this continuous time Markov Chain, we compute
\begin{align*} 
    [\cG f](v) = & \lim_{\del \to 0} \frac{1}{\del} \cdot \l(\E[f(X_{t+\del}) | X_t = v] - f(v) \r) \\
    = & \lim_{\del \to 0} \frac{1}{\del} \cdot \l(1 - \exp(-\del \deg(v))\r) \cdot \frac{1}{\deg(v)} \cdot \l[ \sum_{u \sim v} (f(u) - f(v)) \r] \\ 
    = & \; \deg(v) \cdot \frac{1}{\deg(v)} \cdot \l[ \sum_{u \sim v} (f(u) - f(v)) \r] \\
    = &  \l[ \sum_{u \sim v} (f(u) - f(v)) \r] \\
    = & \; [L f](v),
\end{align*}
where $L$ is the standard (unnormalised) graph Laplacian of $G$.

This completes the proof that the the standard (unnormalised) graph Laplacian $L$ of $G$ is the generator of the continuous time Markov Chain in Algorithm 1.

\subsection{Completely disconnected vs. rarely inter-connected groups} \label{sec:discon-vs-sparscon}

While our theoretical considerations largely focus on the setting where  inter-group connections are absent, for practical purposes, our paradigm is applicable to settings where connections across groups are \textit{rare} but not completely absent; such a scenario being treated as an approximation or a minor deformation of complete disconnection. In the latter setting, it is conceivable that the 0 eigenvalue in the graph Laplacian spectrum has multiplicity only 1; on the other hand there would be a part of the Laplacian spectrum that is very close to 0 but not exactly equal to 0 (for brevity, we will denote it by $\slow$; the full Laplacian spectrum being denoted by $\spec$). Intuitively, this is reflective of the fact that the graph has a group structure that is not fully disconnected, but only rarely connected. If we modified the graph to remove these rare connections across components, these low-lying spectrum of the Laplacian would collapse to 0, and we would be back to the setting of complete disconnection between groups. 

In such a scenario, if the inter-group connections are rare compared to the intra-group connections, we would still expect the rest of the Laplacian spectrum (i.e., $\spec \setminus \slow$) to be well-separated from the above low-lying eigenvalues. As such, our substitute for $\la_g$ would be $\min \l\{ \la : \la \in \spec \setminus \slow       \r\}$. If we denote $\la_\low$  to be $\max \l\{ \la : \la \in \slow   \r\}$, we are operating in the regime where $\la_g \gg \la_\low$. 

In view of these considerations, in the setting of rare but non-zero connections across groups, our heat flow time $\tflow$ needs to be such that $ \tflow \cdot \la_\low $ is small, but $ \tflow \cdot \la_g $ is large. This necessitates a choice of $\tflow$ such that $\frac{1}{\la_g} \ll \tflow \ll \frac{1}{\la_\low}$. Since $\la_\low \ll \la_g$, this enables us to make appropriate choice of the heat flow duration that extend our approach to the setting of rarely connected groups.

\section{Theoretical anlaysis} \label{sec:theory}
In what follows, in the setting of groups / clusters $\{\calC_i\}_{i=1}^K$ we will interchangeably use the notations $T_j=|\calC_j|$ and $T_{\max}=|\calC_{\max}|=\max \{ |\calC_i| : 1\le i \le k \}$.

\subsection{Statements of theoretical results and auxiliary lemmas} \label{sec:results}
\begin{lemma}\label{lem:pen-approx}
    For all $t$ such that
    \[
        G(\beta; t, p) := (p - k) e^{-t\lambda_g} \|\beta^2\|_2 \le \frac{1}{2}\min_{i \in A(\beta)}\frac{\|\beta_{\cC_i}\|_2^2}{T_i},
    \]
    one has
    \[
        |\Lambda_t(\beta) - \Lambda_{\infty}(\beta)| \le G(\beta; t, p)\sum_{i \in A(\beta)} \frac{T_i^{3/2}}{\sqrt{2}\|\beta_{\cC_i}\|_2} + (p - |I(\beta)|) \sqrt{G(\beta; t, p)}.
    \]
    This implies the weaker bound
    \[
        |\Lambda_t(\beta) - \Lambda_{\infty}(\beta)| \le p \sqrt{G(\beta; t, p)}.
    \]
\end{lemma}

For $B[\beta^*; \epsilon] := \{ \beta \mid n^{-1/2}\|X(\beta - \beta^*)\| \le \epsilon\}$, define
\begin{equation}\label{eq:gammadef}
    \Gamma(\beta^*; \epsilon) := \min_{\beta \in \tilde{B}[\beta^*; \epsilon]}  \frac{1}{2\|\beta^2\|_2}\min_{i \in A(\beta)}\frac{\|\beta_{\cC_i}\|_2^2}{T_i}.
\end{equation}
\begin{corollary}\label{thm:closeness-of-losses}
    Let $\epsilon_0 > 0$ be such that $\Gamma(\beta^*; \epsilon_0) > 0$.
    Then for all large enough $t$ such that $(p - k) e^{-t \lambda_g} \le \Gamma(\beta^*; \epsilon_0)$, we have
    \[
        \max_{\beta \in B[\beta^*; \epsilon_0]} |\Lambda_t(\beta) - \Lambda_{\infty}(\beta)| \le C_{\beta^*, \epsilon_0} p \sqrt{p - k} e^{-t\lambda_g/2},
    \]
    where $C_{\beta^*, \epsilon_0} = \max_{\beta \in B[\beta^*; \epsilon_0]} \|\beta^2\|_2$.
\end{corollary}

We set 
\begin{equation} \label{eq:lambda_bounds}
    \Psi_j = \frac{1}{n}X_j^\top X_j  \quad \text{and} \quad  \Lambda(X; \eta) = \max_{1 \le j \le k} \sqrt{\|\Psi_j\|_{\op}}\bigg(1 +  \sqrt{\frac{4\log \eta^{-1}}{T_j}}\bigg).
\end{equation}
Then we have: 
\begin{lemma}\label{lem:approx-min}
    Let $\lambda \ge \frac{\sigma}{\sqrt{n}} \Lambda(X; \eta)$, and, for $0 < \epsilon \le \epsilon_0$, let $\hat{\beta}_{t, \lambda} = \hat{\beta}_{t, \lambda}^{(\epsilon)}$ be the minimiser of $F_{t, \lambda}(\beta)$ in $B[\beta^*, \epsilon]$. Assume also that $8 T_{\max} \lambda \|\beta^*\|_{2, 1} \le \epsilon$. Then, with probability at least $1 - 2k\eta$, we have that $\hat{\beta}_{t, \lambda}$ is an approximate minimiser of the group Lasso objective $F_{\infty, \lambda}$ in the sense that
    \[
        F_{\infty, \lambda}(\hat{\beta}_{t, \lambda}) \le \min_{\beta} F_{\infty, \lambda}(\beta) + 2 C_{\beta^*, \epsilon_0} p \sqrt{p - k} e^{-t\lambda_g/2}.
    \]
\end{lemma}
Using Lemma~\ref{lem:approx-min}, we can prove an approximate sparsity oracle inequality for $\hat{\beta}_{t, \lambda}$.
\begin{lemma}\label{lem:approx-oracle}
    Under the assumptions of Lemma~\ref{lem:approx-min}, we have with probability at least $1 - 2k\eta$ that
    \begin{align}\label{eq:approx-sparsity-oracle} \nonumber
        \frac{1}{2n}\|X(\hat{\beta}_{t, \lambda} - \beta^*)\|_2^2 &+ \lambda \sum_j \sqrt{T_j} \|\hat{\beta}_{t, \lambda}^j - \beta^j\|_2 \\
                                                                  &\le \frac{1}{2n} \|X(\beta - \beta^*)\|_2^2 + 4 \lambda \sum_{j \in A(\beta)} \sqrt{T_j} \min \{\|\beta^j\|_2 \|\hat{\beta}_{t, \lambda}^j - \beta^j\|_2\} + E,
    \end{align}
    where $E = 2 C_{\beta^*, \epsilon_0} p \sqrt{p - k} e^{-t\lambda_g/2}$.
\end{lemma}
The approximate sparsity oracle inequality will give us a prediction consistency result. Without any further assumptions we have a slow-rate result. For getting faster rates we assume a restricted eigenvalue property.

We now re-state our main theorem from the main text.
\begin{theorem}\label{thm:prediction-consistency-app}
    Let $\lambda \ge \frac{\sigma}{\sqrt{n}} \Lambda(X; \eta)$, and, for $0 < \epsilon \le \epsilon_0$, let $\hat{\beta}_{t, \lambda} = \hat{\beta}_{t, \lambda}^{(\epsilon)}$ be the minimiser of $F_{t, \lambda}(\beta)$ in $B[\beta^*, \epsilon]$. Let $8 |\cmax| \lambda \|\beta^*\|_{2, 1} \le \epsilon$.
    Then  with probability at least $1 - 2k\eta$, we have
    \begin{equation} \label{eq:PC_slow-app}
        \frac{1}{n}\|X(\hat{\beta}_{t, \lambda} - \beta^*)\|_2^2 = O(\|\beta^*\|_{2, 1} \lambda |\cmax| + p^{3/2} e^{-t\lambda_g/2}).
    \end{equation}    
    If we further assume $\RE(s)$ with parameter $\kappa$, then
    \begin{equation} \label{eq:PC_fast-app}
        \frac{1}{n}\|X(\hat{\beta}_{t, \lambda} - \beta^*)\|_2^2 = O\bigg(\frac{s\lambda^2 |\cmax|}{\kappa^2} + p^{3/2} e^{-t\lambda_g/2}\bigg).
    \end{equation}
\end{theorem} 

We further demonstrate that the RE($s$) property above for the random design matrix $X$ can be related, in the crucial setting of Gaussian covariates, to a similar property for the \textit{deterministic} covariance matrix $\Sigma$ of the rows of $X$. The latter can often be much easier to verify; for instance, it may be easily seen to hold as soon as $\Sigma$ is well-conditioned. This is captured in the following lemma, re-stated from the main text, along with an explicit  numerical lower bound.

\begin{lemma}\label{lem:res-eig-app}
    Suppose that $\Sigma$ satisfies $\RE(s)$ with parameter $\kappa_{\Sigma}(s) > 0$. Let the rows of $X$ be i.i.d. $N(0, \Sigma)$. Assume that
    \begin{equation}
        n \ge (36 \cdot 8)^2 \frac{(\rho(\Sigma))^2 s |\cmax| \log p}{(\kappa_{\Sigma}(s))^2}.
    \end{equation}
    Then $\frac{X^\top X}{n}$ satisfies $\RE(s)$ with parameter $\kappa_{\Sigma}(s)/8$ with probability at least $1 - c'\exp(-cn)$ for some constants $c, c' > 0$.
\end{lemma}

\subsection{Proofs of the theoretical results}\label{sec:proof}
\begin{proof}[Proof of Lemma~\ref{lem:pen-approx}]
    We begin with the decomposition
    \[
        e^{-tL}(\beta \odot \beta) = \sum_i e^{-t\lambda_i} \langle v_i, \beta^2\rangle v_i. 
    \]
    For the eigenspace corresponding to $0$, we choose the basis $\{\mathbf{1}_{\cC_i} / \sqrt{|\cC_i|}\}_{i = 1}^k$. Thus we may write
    \[
        e^{-tL}(\beta \odot \beta) = g + \xi,
    \]
    where
    \[
        g = \sum_{i = 1}^k \frac{\|\beta_{\cC_i}\|_2^2}{|\cC_i|}\mathbf{1}_{\cC_i} \quad \text{ and } \quad \xi = \sum_{i > k} e^{-t\lambda_i} \langle v_i, \beta^2\rangle v_i.
    \]
    Also note that $\Lambda_{\infty}(\beta) = \langle \Phi(g), \mathbf{1} \rangle$. Therefore
    \begin{align*}
        |\Lambda_t(\beta) & - \Lambda_{\infty}(\beta)|                          \\
                          & = |\langle \Phi(g + \xi) - \Phi(g), \mathbf{1}\rangle | \\
                          & \le \sum_{\ell = 1}^p |\sqrt{|g_\ell + \xi_\ell|} - \sqrt{|g_\ell|}|  \\
                          & \le \sum_{\ell \in I(\beta)} \frac{|\xi_\ell|}{2\sqrt{\theta_\ell}} + \sum_{\ell \notin I(\beta)} \sqrt{|\xi_\ell|} \quad (\text{where }\theta_\ell = \alpha_\ell |g_\ell + \xi_\ell| + (1 - \alpha_\ell) |g_\ell|, \alpha_\ell \in [0, 1])
    \end{align*}
    Now
    \[
        \xi_\ell = \sum_{i > k} e^{-t\lambda_i} \langle v_i, \beta^2 \rangle v_{i, \ell}.
    \]
    We have the following straightforward uniform bound
    \[
        |\xi_\ell| \le \sum_{i > k} e^{-t \lambda_g} |\langle v_i, \beta^2 \rangle| \le (p - k) e^{-t\lambda_g} \|\beta^2\|_2 = G(\beta; t, p).
    \]
    Now
    \[
        g_{\ell} = \sum_{i = 1}^k \frac{\|\beta_{\cC_i}\|_2^2}{|\cC_i|}\mathbf{1}_{\cC_i}(\ell).
    \]
    Assume $t$ is such that
    \[
        G(\beta; t, p) \le \frac{1}{2}\min_{i \in A(\beta)}\frac{\|\beta_{\cC_i}\|_2^2}{|\cC_i|}.
    \]
    Then, clearly, $\theta_\ell \ge g_\ell/2$ for all $\ell \in I(\beta)$. It follows that
    \[
        |\Lambda_t(\beta) - \Lambda_{\infty}(\beta)| \le G(\beta; t, p) \sum_{\ell \in I(\beta)} \frac{1}{\sqrt{2 g_\ell}} + (p - |I(\beta)|) \sqrt{G(\beta; t, p)}.
    \]
    Now note that
    \begin{align*}
        \sum_{\ell \in I(\beta)} \frac{1}{\sqrt{g_\ell}} &= \sum_{\ell \in I(\beta)} \sum_{i \in A(\beta)} \frac{\sqrt{|\cC_i|}}{\|\beta_{\cC_i}\|_2}\mathbf{1}_{\cC_i}(\ell) \\
                                                         &= \sum_{i \in A(\beta)} \frac{|\cC_i|^{3/2}}{\|\beta_{\cC_i}\|_2}.
    \end{align*}
    This completes the proof of the first bound. Now, under our assumptions on $t$, we have
    \begin{align*}
        \sum_{i \in A(\beta)} \frac{|\cC_i|^{3/2}}{\|\beta_{\cC_i}\|_2}                   &\le \sum_{i \in A(\beta)} \frac{|\cC_i|}{\sqrt{2 G(\beta; t, p)}} \\
                                                                                          &= \frac{|I(\beta)|}{\sqrt{2 G(\beta; t, p)}}.
    \end{align*}
    Combining this with the first bound, we get
    \[
        |\Lambda_t(\beta) - \Lambda_{\infty}(\beta)| \le (p - |I(\beta)|/2) \sqrt{G(\beta; t, p)} \le p \sqrt{G(\beta; t, p)}. 
    \]
    This completes the proof.
\end{proof}

\begin{proof}[Proof of Lemma~\ref{lem:approx-min}]
    For our choice of $\lambda$, any group Lasso solution $\hat{\beta}_{\infty, \lambda}$ satisfies (see Eq. (3.9) of \cite{lounici2011oracle})
    \[
        \frac{1}{n}\|X(\hat{\beta}_{\infty, \lambda} - \beta^*)\|_{2, 1} \le 8 \lambda \sqrt{T_{\max}} \|\beta^*\|_{2, 1}.
    \]
    Thus under our assumptions, a group Lasso solution $\hat{\beta}_{\infty, \lambda}$ will lie inside $B[\beta^*; \epsilon]$. Now we have
    \begin{align*}
        F_{\infty, \lambda}(\hat{\beta}_{t, \lambda}) &= F_{t, \lambda}(\hat{\beta}_{t, \lambda}) + \Lambda_\infty(\hat{\beta}_{t, \lambda}) - \Lambda_t(\hat{\beta}_{\lambda}) \\
                                                      &\le F_{t, \lambda}(\hat{\beta}_{\infty, \lambda}) + C_{\beta^*, \epsilon_0} p \sqrt{p - k} e^{-t\lambda_g/2} \\
                                                      &= F_{\infty, \lambda}(\hat{\beta}_{\infty, \lambda}) + \Lambda_t(\hat{\beta}_{\infty, \lambda}) - \Lambda_\infty(\hat{\beta}_{\infty, \lambda}) + C_{\beta^*, \epsilon_0} p \sqrt{p - k} e^{-t\lambda_g/2} \\
                                                      &\le F_{\infty, \lambda}(\hat{\beta}_{\infty, \lambda}) + 2 C_{\beta^*, \epsilon_0} p \sqrt{p - k} e^{-t\lambda_g/2}.
    \end{align*}
    This completes the proof.
\end{proof}
\begin{proof}[Proof of Lemma~\ref{lem:approx-oracle}]
    By Lemma~\ref{lem:approx-min}, for any $\beta$,
    \begin{equation}\label{eq:approx-basic}
        F_{\infty, \lambda}(\hat{\beta}_{t, \lambda}) \le F_{\infty, \lambda}(\beta) + E.
    \end{equation}
    This may be thought of as an approximate ``basic inequality'' for the estimator $\hat{\beta}_{t, \lambda}$. Now we use the arguments used in the proof of the sparsity oracle inequality for group Lasso in \cite{lounici2011oracle}. We can rewrite \eqref{eq:approx-basic} as
    \begin{align*}
        \frac{1}{2n}\|X(\hat{\beta}_{t, \lambda} &- \beta^*)\|_2^2\\ &\le \frac{1}{2n}\|X(\beta - \beta^*)\|_2^2 + \frac{1}{n} \varepsilon^\top X(\hat{\beta}_{t, \lambda} - \beta) + \lambda \sum_{j} \sqrt{T_j} (\|\hat{\beta}^j\|_2 - \|\hat{\beta}_{t, \lambda}^j\|_2) + E. 
    \end{align*}
    By Cauchy-Schwartz, we have
    \[
        \varepsilon^\top X (\hat{\beta}_{t, \lambda} - \beta) \le \sum_j \|\varepsilon^\top X_j\|_2 \|\hat{\beta}_{t, \lambda}^j - \beta^j\|_2.
    \]
    Consider the events $\mathcal{A}_j = \{n^{-1}\|\varepsilon^\top X_j\|_2 \le \lambda \sqrt{T_j}\}$. Since $n^{-1}X_j^\top\varepsilon \sim N\bigg(0, \sigma^2 \frac{X_j^\top X_j}{n}\bigg)$, we have using Lemma B.1 of \cite{lounici2011oracle} that
    \[
        \P(\mathcal{A}_j^c) \le 2\eta
    \]
    provided
    \[
        \lambda \ge \frac{\sigma}{\sqrt{n}} \sqrt{\frac{1}{T_j}(\tr(\Psi_j) + 2 \|\Psi_j\|_{\op} (2 \log \eta^{-1} + \sqrt{T_j \log \eta^{-1}}))}.
    \]
    A simpler sufficient condition for this is
    \[
        \lambda \ge \frac{\sigma}{\sqrt{n}} \sqrt{\|\Psi_j\|_{\op}}\bigg(1 +  \sqrt{\frac{4\log \eta^{-1}}{T_j}}\bigg).
    \]
    Let
    \[
        \Lambda(X; \eta) = \max_{1 \le j \le k} \sqrt{\|\Psi_j\|_{\op}}\bigg(1 +  \sqrt{\frac{4\log \eta^{-1}}{T_j}}\bigg).
    \]
    Thus, if $\lambda \ge \frac{\sigma}{\sqrt{n}}\Lambda(X; \eta)$, then with probability at least $1 - 2k\eta$, we have
    \[
        \frac{1}{n}\varepsilon^\top X (\hat{\beta}_{t, \lambda} - \beta) \le \lambda \sum_j \sqrt{T_j} \|\hat{\beta}_{t, \lambda}^j - \beta^j\|_2
    \]
    Combining we get that with probability at least $1 - 2k\eta$, we have
    \begin{align*}
        \frac{1}{2n}\|X(\hat{\beta}_{t, \lambda} - \beta^*)\|_2^2 &+ \lambda \sum_j \sqrt{T_j} \|\hat{\beta}_{t, \lambda}^j - \beta^j\|_2 \\
                                                                  &\le \frac{1}{2n} \|X(\beta - \beta^*)\|_2^2 + 4 \lambda \sum_{j \in A(\beta)} \sqrt{T_j} \min \{\|\beta^j\|_2 \|\hat{\beta}_{t, \lambda}^j - \beta^j\|_2\} + E.
    \end{align*}
    This is the desired approximate sparsity oracle inequality \eqref{eq:approx-sparsity-oracle}.
\end{proof}

\begin{proof}[Proof of Theorem~\ref{thm:prediction-consistency-app}]
    The slow-rate prediction consistency result follows immediately from \eqref{eq:approx-sparsity-oracle} by taking $\beta = \beta^*$. The fast-rate result follows from the $\RE(s)$ assumption as in proof of Theorem 3.1 in \cite{lounici2011oracle}.
\end{proof}

\begin{proof}[Proof of Lemma~\ref{lem:res-eig-app}]
    By Cauchy-Schwartz, we have
    \[
        \sum_{i \in G_j} |\Delta_i| \le \sqrt{T_j} \|\Delta^j\|.
    \]
    Thus
    \begin{align*}
        \sum_{j \notin A}\sum_{i \in G_j} |\Delta_i| &\le \sum_{j \notin A} \sqrt{T_j} \|\Delta^j\| \\
                                                     &\le 3 \sum_{j \in A} \sqrt{T_j} \|\Delta^j\|. 
    \end{align*}
    Hence
    \[
        \|\Delta\|_1 \le 4 \sum_{j \in A} \sqrt{T_j} \|\Delta^j\|.
    \]
    Another application of Cauchy-Schwartz gives
    \[
        \|\Delta\|_1 \le 4 \sqrt{\sum_{j \in A} T_j} \|\Delta_A\| \le 4 \sqrt{s T_{\max}} \|\Delta_A\|.
    \]
    Now using Theorem 1 of \cite{raskutti10a} we get that with probability at least $1 - c'\exp(-cn)$,
    \begin{align*}
        \frac{\|X \Delta\|}{\sqrt{n}} &\ge \frac{1}{4}\|\Sigma^{1/2}\Delta\| - 9 \rho(\Sigma) \sqrt{\frac{\log p}{n}}\|\Delta\|_1 \\
                                      &\ge \bigg(\frac{\kappa_{\Sigma}(s)}{4} - 36 \rho(\Sigma) \sqrt{\frac{s T_{\max}\log p}{n}}\bigg)\|\Delta_A\| \\
                                      &\ge \frac{\kappa_{\Sigma}(s)}{8} \|\Delta_A\|,
    \end{align*}
    provided
    \[
        n \ge (36 \cdot 8)^2 \frac{(\rho(\Sigma))^2 s T_{\max} \log p}{(\kappa_{\Sigma}(s))^2}.
    \]
    This completes the proof.
\end{proof}

\section{Analysis of sample complexity and prediction guarantees for random designs}    \label{sec:analysis_models}

\subsection{Prediction guarantees and sample complexity bounds}  \label{sec:recovery-app}
In this section, we will demonstrate quantitative  guarantees for the prediction error and sample complexity in our heat flow based approach for key models of random designs with a group structure. These include, in particular, GFFs on typical clustered networks and Gaussian designs based on SBMs. To this end, we appeal to Theorem \ref{thm:prediction-consistency-app}. As we shall see below, both settings satisfy the RE($s$) property, so via Theorem \ref{thm:prediction-consistency-app} and Lemma \ref{lem:res-eig-app}, we will obtain concrete prediction error guarantees as well as bounds on the sample complexity as soon as we can bound the quantities $\la, \rho(\S)$ and $\kappa_\S(s)$. 

We first turn our attention to the quantity $\la$. To this end, we invoke \eqref{eq:lambda_bounds} and Lemma \ref{lem:approx-min}.  In the present article, we will content ourselves with a polynomial decay of probability, which implies that the quantity $\eta=O(n^{-\a})$ for some $\a$. Since the maximal group size $|\calC_{\max}| \ge 1$, we have $\l( 1+ \sqrt{4 \log \eta^{-1}}{|\calC_{\max}|}| \r)=O(\sqrt{\log n})$.  On the other hand, we have the bound $\max_j \|\Psi_j\|_{\mathrm op}=\max_j \|\frac{1}{n}X_j^TX_j\|_{\mathrm op}=O_P(\s_{\max}(\S))$, where $\s_{\max}(\S)$ is the maximal singular value of the population covariance matrix $\S$.    \eqref{eq:lambda_bounds} therefore implies that $\L(X;\eta)=O_P(\s_{\max}(\S)\sqrt{\log n})$. As a result, Lemma \ref{lem:approx-min} suggests that we consider $\la \gtrsim  \s_{\max}(\S) \sqrt{\frac{\log n}{n}}$.

We next observe that if $\s_{\max}(\S)$ and $\s_{\min}(\S)$ are respectively the maximum and minimum singular values of the population covariance matrix $\S$, then $\rho(\S) = \max_i \S_{ii} \le \s_{\max}(\S)$ and $\kappa_\S(s) \ge \s_{\min}(\S)$, which are direct consequences of the definitions of the quantities in question. In view of  
Theorem \ref{thm:prediction-consistency-app} and, in particular \eqref{eq:PC_fast-app}, this leads to a somewhat simplified prediction guarantee of 
\begin{equation} \label{eq:PG-app}
    \frac{1}{n}\|X(\hat{\beta}_{t, \lambda} - \beta^*)\|_2^2 = O_P\l(s\; |\calC_{\max}| \cdot \frac{\s_{\max}(\S)}{\s_{\min}(\S)^2} \cdot \frac{\log n}{n} + p^{3/2} e^{-t\lambda_g/2}  \r) . 
\end{equation} 
Our goal here is to understand the order of the time $\tflow$ and the step count $\Ns$ (in terms of the other parameters of the problem) till which we need to run our heat flow based algorithm in order to achieve a desired accuracy.  We will  bifurcate our analysis into two related sections.

\subsubsection{Bounds on \texorpdfstring{$\tflow$}{} and \texorpdfstring{$\Ns$}{} for given \texorpdfstring{$n, p$}{}}

First, given $n,p$, we will demonstrate the order of $\tflow$ at which the approximation error due to our heat flow based approach (roughly, the second term in \eqref{eq:PG-app}) becomes comparable to  the contribution to the prediction error bound from the classical group lasso methods that assume complete knowledge of the group structure (roughly, the first term in \eqref{eq:PG-app}). Heuristically, this indicates the time we need to run the heat flow in order to be comparable to the classical group lasso (but without requiring complete knowledge of the groups, unlike the classical setting).  Equating the two terms in \eqref{eq:PG-app}, we deduce that it suffices to take
\begin{equation} \label{eq:time_10-app}
    \tflow ~\gtrsim \frac{1}{\la_g} \log p + \frac{1}{\la_g} \log \l( \frac{1}{s\; |\calC_{\max}|} \cdot  \frac{ \s_{\min}(\S)^2} {s \s_{\max}(\S)} \cdot \frac{n}{\log n}\r). 
\end{equation}
In most settings of interest, we have bounds of the form $ p^{-a} \lesssim \s_{\min}(\S) \le  \s_{\max}(\S) \lesssim p^b$ for some $a,b \ge 0$. In particular, this holds for the GFF and block model strucutred covariates that we discuss in the present work. Further, we have the trivial  bounds entailing $s,|\calC_{\max}| \in [1,p]$. Combining these observations with \eqref{eq:time_10-app}, we deduce that for such models, we have the much simpler prescription 
\begin{equation} \label{eq:time_11}
    \tflow ~\gtrsim \frac{1}{\la_g} \max \{ \log p , \log n \}. 
\end{equation}

As discussed in the main text, this implies it suffices to have $\Ns=O\l(d_{\max} \cdot \frac{1}{\la_g} \cdot \max \{ \log p , \log n \}\r)$.

For a definite quantitative ballpark for $\tflow$ and $\Ns$, we focus on the setting of typical clustered networks (c.f. Sec. \ref{sec:graph_struct}.). In such a setting, we may deduce that $d_{\max}=\Theta_P(p)$ whereas $\la_g=\Theta_P(p)$; for details we refer the reader to Sec. \ref{sec:graph_struct}.

This implies that we can further simplify to the prescriptions with high probability
\begin{equation}
    \tflow ~\gtrsim \frac{1}{p} \cdot \max \{ \log p , \log n \}
\end{equation}
and 
\begin{equation}
    \Ns =O_P\l(\max \{ \log p , \log n \}\r).
\end{equation}

\subsubsection{Bounds on \texorpdfstring{$n, \tflow, \Ns$}{} for given prediction guarantee \texorpdfstring{$\eps$}{}}

Herein, we fix a prediction error guarantee $\eps$, and make explicit prescriptions for the order of $n, \tflow$ and $\Ns$ that will allow us to obtain a prediction error of order $O(\eps)$. To this end, we posit that the two terms on the right hand side of  \eqref{eq:PG-app} are separately $O(\eps^2)$ (since the left hand side is the squared prediction error). For $\tflow$, this entails that $p^{3/2}e^{-\la_g \tflow} \lesssim \eps^2$, which translates into 
\begin{equation} \label{eq:time_2}
    \tflow \gtrsim \frac{1}{\la_g}\l( \log p + \log \frac{1}{\eps} \r). 
\end{equation}

In the setting for typical clustered networks (c.f. Sec. \ref{sec:graph_struct}), we may deduce that $d_{\max}=\Theta_P(p)$ whereas $\la_g=\Theta_P(p)$ , implying  that we can further simplify  to the following bounds that hold with high probability:
\begin{equation} 
    \tflow ~\gtrsim \frac{1}{p} \cdot \max \l\{ \log p , \log \frac{1}{\eps} \r\}; \quad \Ns =O_P\l(\max \l\{ \log p , \log \frac{1}{\eps} \r\}\r).
\end{equation}

For prescribing $n$ under a target prediction error $\eps$, we need to satisfy two conditions: 

(a) the first term on the right  in \eqref{eq:PG-app} is $O(\eps^2)$, which leads to 
\[    \frac{n}{\log n} \gtrsim ~\frac{1}{\eps^2}  \l( s\; |\calC_{\max}| \cdot  \frac{ \s_{\max}(\S)  }{ \s_{\min}(\S)^2}\r).   \]

(b) As per Lemma \ref{lem:res-eig-app} and the bound $\kappa_{\Sigma}(s) \ge \s_{\min}(\S)$, we have 
\[n \gtrsim s \; |\calC_{\max}| \cdot \frac{(\rho(\Sigma))^2}{(\s_{\min}(\S))^2} \cdot \log p. \]

Combining the last two bounds, we obtain the unified bound 
\begin{equation} \label{eq:n_eps-app}
    \frac{n}{\log n} \gtrsim\max \l\{  \l( \frac{1}{\eps^2} \cdot s |\calC_{\max}| \cdot \frac{ \s_{\max}(\S)  }{ \s_{\min}(\S)^2}\r), \l( s \; |\calC_{\max}| \cdot \frac{(\rho(\Sigma))^2}{(\s_{\min}(\S))^2} \cdot \log p \r) \r\}.
\end{equation}

\subsection{On the structure of typical clustered networks} \label{sec:graph_struct}

In this section, we explore the structure of the \textit{typical clustered network}  $G$ on $p$ vertices and with $k$ clusters, where $k=O(1)$, and the size of each cluster is $\Theta(p)$. We model a \textit{typical clustered network} with these parameters as follows. We posit that the clusters $\{\calC_i\}_{i=1}^k$ are fully disconnected across clusters; thus the graph $G$ has exactly $k$ connected components given by the $\calC_i$-s. Each component $\calC_i$ is modelled as a \textit{dense} random graph with $|\calC_i|$ vertices and edge connection probability $\xi_i \in (0,1)$. For the sake of definiteness, we allow self-loops in our model; this would not exert any major influence on the large scale properties of the graph. Since the number of components $k$ is $O(1)$, we may take the $\xi_i$-s to be bounded away from 0 and 1, as the parameters $n$ and $p$ grow.

\subsubsection{Maximum and minimum degrees of typical clustered networks} \label{sec:TCN-deg}

We first demonstrate that the maximum and minimum degrees of $G$, denoted resp. $d_{\max}$ and $d_{\min}$, are both $\Theta_P(p)$. To this end, we observe that if $d_{\max,i}$ and $d_{\min,i}$ are resp. the maximum and minimum degrees of $\calC_i$, then $d_{\max}=\max_{1 \le i \le k} d_{\max,i}$ and $d_{\min}=\min_{1 \le i \le k} d_{\min,i}$.

For any particular vertex $v \in \calC_i$, its degree $\deg(v)$ is distributed as Binomial($T_i,\xi_i,$), where we recall the notation that $T_i=|\calC_i|$. Clearly, $\E[\deg(v)]=T_i \xi_i$. By a well-known  large deviation estimate (c.f. \cite{dembo_zeitouni_2010}), for any $\del>0$, we have 
\begin{equation}
    \P \l[ (1-\del) T_i \xi_i  \le  \deg(v)  \le (1+\del) T_i \xi_i  \r] \ge 1 - \exp(-C(\del,\xi_i)T_i),
\end{equation}
where $0<C(\del,\xi_i)<\infty$ is a quantity that depends only on $\del$ and $\xi_i$. By a union bound, this implies that 
\begin{equation}
    \P \l[ \forall \; v \in \calC_i, \; \text{it holds that} \; (1-\del) T_i \xi_i  \le  \deg(v)  \le (1+\del) T_i \xi_i  \r] \ge 1 - T_i \cdot \exp(-C(\del,\xi_i)T_i).
\end{equation}

But \[  \forall \; v \in \calC_i, \; \text{it holds that} \; (1-\del) T_i \xi_i  \le  \deg(v)  \le (1+\del) T_i \xi_i  \equiv (1-\del) T_i \xi_i  \le  d_{\min,i} \le d_{\max_i}  \le (1+\del) T_i \xi_i.   \] Since $d_{\max}=\max_{1 \le i \le k} d_{\max,i}$ and $d_{\min}=\min_{1 \le i \le k} d_{\min,i}$, by a further union bound we may deduce that 
\begin{equation}
    \P \l[ (1-\del) \min_{1\le i \le k} (T_i \xi_i)  \le  d_{\min,i} \le d_{\max_i}  \le (1+\del) \max_{1\le i \le k}(T_i \xi_i)  \r] \ge 1 -  \l(\sum_{i=1}^k T_i \cdot \exp(-C(\del,\xi_i)T_i) \r).
\end{equation}

Since in our model of typical clustered networks, $k=O(1)$ while each $T_i=\Theta(p)$ and $\xi_i$ are bounded away from 0 and 1, we may deduce that for any $\del>0$ we have constants $C_1,C_2$ and $C\l(\del,\{\xi_i\}_{i=1}^k\r)$
\begin{equation}
    \P \l[ (1-\del) C_ 1 \cdot p \le  d_{\min} \le d_{\max}  \le (1+\del) C_2 \cdot p \r] \ge 1 -  \exp\l(-C\l(\del,\{\xi_i\}_{i=1}^k\r)p\r).
\end{equation}

This demonstrates that, for a typical clustered network on $p$ nodes, we have $d_{\max}=\Theta_P(p)$ as well as $d_{\min}=\Theta_P(p)$.

\subsubsection{Spectral gap for typical clustered networks} \label{sec:TCN-spec_gap}

In this section, we investigate the order of the spectral gap for typical clustered networks, as defined above. To this end, we invoke results of Chung-Lu-Vu on spectra of random graphs with given expected degree \citep{ChungPNAS,chung2004spectra}. To state their results, we define the random graph model with a given expected degree sequence as follows. 

For a graph on $p$ nodes and a given sequence of non-negative reals $\mathbf{w}=(w_1,\ldots,w_p)$ satisfying $\max_i w_i^2 \le \sum_{k=1}^p w_k$, we define the random graph with the expected degree sequence \citep{ChungPNAS,chung2004spectra} by connecting the vertices labelled $i$ and $j$ with an edge with probability $w_iw_j \rho$, where $\rho=(\sum_{k=1}^p w_k)^{-1}$. A classical random graph $G(p,\xi)$ with $p$ vertices and edge connection probabilities $\xi$ is obtained in the above model by choosing $w_i=p\xi$. 

Two quantities of relevance in the results of \cite{ChungPNAS,chung2004spectra} are the expected average degree $\md=\frac{1}{n} \sum_{k=1}^p w_k$ and the second order average degree  $\snd = \l(\sum_{k=1}^p w_k^2\r) \big/ \l(\sum_{k=1}^p w_k\r)$; we also set $w_{\max}:=\max_k w_k$ and $w_{\min}:=\min_k w_k$. For the $G(p,\xi)$ model considered above, we have $\snd=\md=w_{\max}=w_{\min}=p\xi$. 

For the above random graph model, if the normalised Laplacian $L^* = D^{-1/2} L D^{-1/2}$ (where $L$ is the usual, unnormalised graph Laplacian and $D$ is the diagonal matrix of degrees) has spectrum $0 = \la_0(L^*) \le \la_0(L^*) \le \ldots \le \la_{p-1}(L^*)$, then \cite{ChungPNAS} demonstrates that with high probability we have the bound $\max_{i \ne 0}|1 - \la_i(L^*)| \lesssim \frac{1}{\sqrt{\md}} + o(\frac{\log^3 p}{w_{\min}})$. In our setting of interest, namely the $G(p,\xi)$ random graph, this implies that with high probability we have $\la_1(L^*) \gtrsim 1 - \frac{1}{\sqrt{p \xi}} (1+o(1))$.

We now observe that $\la_1(L)=\la_1(D^{1/2}L^* D^{1/2}) \ge d_{\min}\la_1(L^*)$. But we have already demonstrated that for a $G(p,\xi)$ random graph, we have $d_{\min}=\Theta_P(p)$, which when combined with the analysis above yields $\la_1(L) \gtrsim p$ with high probability for such graphs.

It remains to note that the ground state (equiv., the lowest non-zero eigenvalue of the unnormalised graph Laplacian) $\la_g$ for the \textit{typical clustered network} is given by $\la_g = \min_{1 \le i \le k} \la_1(L_{\calC_i})$, where $L_{\calC_i}$ is the unnormalised graph Laplacian corresponding to the component subgraph $\calC_i$. But in a typical clustered network, each $\calC_i$ is a $G(p,\xi_i)$ random graph for some $\xi_i \in (0,1)$. Since the number of components $k$ is $O(1)$, this leads to the bound $\la_g \gtrsim p$ with high probability.

Finally, we observe via \eqref{eq:sparse_eig_bound} that $\la_g=\la_1(L)\le \la_{\max}(L) \le 2 d_{\max}= \Theta_P(p)$. This, in particular, implies that for typical clustered networks we have $\la_g = \Theta_P(p)$, as desired.  

\subsection{Sample complexity and prediction guarantees for covariate structures based on GFF} \label{sec:recovery_GFF}

\subsubsection{The order of the mass parameter \texorpdfstring{$\t$}{}}
A few words are in order regarding the size of the mass parameter (or convexity parameter) $\t$.
Our motivation behind introducing the parameter $\t$ is to create sufficient convexity to overcome the singular nature of $L$, without changing the essential orders of magnitude associated with the model. 
For the upper bound on the spectrum of the covariance matrix in the GFF case, we  will proceed via the identity $\s_{\max}(\S)=(\s_{\min}(L)+\t)^{-1}$. Since $L$ is singular,  $\s_{\min}(L)=0$; thus $\s_{\max}(\S)=\t^{-1}$.  In view of this, if $\s_{\text{low}}(L)$ is the smallest non-zero eigenvalue of $L$, we will select $\t$ to be simply equal to $\s_{\mathrm{low}}(L)$, thereby adding a minimal amount of convexity without making essential changes to its large scale behaviour. For the spectral lower bound on the covariance structure of GFFs, we clearly have $\s_{\min}(\S)=(\s_{\max}(L)+\t)^{-1}$.

\subsubsection{Spectral bounds}

The study of extremal eigenvalues of graph Laplacians has a long history in spectral graph theory; for a comprehensive account we refer the reader to the surveys \cite{spielman2007spectral,spielman2012spectral}. In this work, we content ourselves with the following general bounds, which are essentially versions of results known in the literature, or follow via simple considerations therefrom. 

\begin{lemma} \label{lem:graph_spectra}
    For $G=(V,E)$ be a connected graph, we denote by $L$ and $L^*$ the unnormalised and normalised graph Laplacians respectively. Then the following general spectral bounds hold.

    We have the spectral upper bounds
    \begin{equation} \label{eq:dense_eig_bound}
        \s_{\max}(L) \le \frac{\sqrt{1+8|E|}-1}{2}
    \end{equation}
    and
    \begin{equation} \label{eq:sparse_eig_bound}
        \s_{\max}(L) \le 2 d_{\max},
    \end{equation}
    where $d_{\max}$ is the maximum degree of the graph $G$.

    On the other hand, we have the spectral lower bound
    \begin{equation} \label{eq:Cheeger_eig_bound}
        \s_{\mathrm{low}}(L^*) \ge \frac{1}{2}\mathfrak{K}(G)^2,
    \end{equation}
    where $\s_{\mathrm{low}}$ is the smallest non-zero eigenvalue of $L^*$, $\mathfrak{K}(G)$ is the conductance of the graph $G$, defined by \[\mathfrak{K}(G):=\min_{S \subset V, S \ne \phi} \frac{E(S,V \setminus S)}{\min\{Vol(S),Vol(V \setminus S)\}}\] with $E(S,V\setminus S)$ being the number of edges between the vertices in $S$ and $V \setminus S$, and $Vol(A)$ for $A \subseteq V$ being defined as $Vol(A)=\sum_{v \in A} \deg(v)$. 
\end{lemma}

We observe that in our case, the graph underlying the GFF is not connected. We will, therefore, apply Lemma \ref{lem:graph_spectra} to each connected component (i.e., group) of the vertices. The component-wise bounds can then be combined to obtain \[  \min_i \s_{\min}(\S_{\calC_i})  \le \s_{\min}(\S) \le \s_{\max}(\S) \le \max_i \s_{\max}(\S_{\calC_i}). \] We will invoke Lemma \ref{lem:graph_spectra} for each connected component of the underlying graph.

\subsubsection{Random designs on typical clustered networks}

For the GFF on typical clustered networks, we invoke the inequality \eqref{eq:sparse_eig_bound}. As a result, we have $\s_{\max}(L_{\calC_i}) \lesssim \d_{\max}(\calC_i)$. It has been demonstrated in Section \ref{sec:graph_struct} that the last quantity is, with high probability, $\Theta_P(p)$.
We therefore have, with high probability, the bound
\begin{equation} \label{eq:GFF_dense-1}
    \s_{\min}(\S) \simeq \s_{\max}(L)^{-1} = \Theta(p^{-1}). 
\end{equation}

For analysing the behaviour of $\s_{\max}(\S)$, equivalently that of $\s_{\low}(L)$ (which is the smallest non-zero eigenvalue of $L$), we invoke the analysis in Section \ref{sec:graph_struct} and deduce that $\s_{\low}(L)=\la_g=\Theta(p)$ with high probability. 

We therefore have, with high probability, that
\begin{equation} \label{eq:GFF_dense-2}
    \rho(\S) \le  \s_{\max}(\S) \simeq \s_{\low}(L)^{-1}  = O(p^{-1}). 
\end{equation} 

For typical clustered networks, the group sizes are comparable to each other (i.e., their ratios are uniformly bounded in $p,n$), therefore recalling that $s$ is the number of non-vanishing groups  we roughly have $s |\calC_{\max}| \simeq \| \b^* \|_0$, where $ \| \b^* \|_0$ is the $L^0$ norm (equivalently, the support size) of the true signal $\b^*$.
Applying \eqref{eq:GFF_dense-1} and \eqref{eq:GFF_dense-2} to \eqref{eq:n_eps-app}, we may therefore deduce that 
\[ 
    \l( \frac{n}{\log n} \r)_{\text{GFF}}  \gtrsim \max \l\{ \frac{1}{\eps^2} \cdot \|\b^*\|_0 \cdot  p,~~ \|\b^*\|_0 \cdot \log p \r\}. 
\] 
It may be noted that the $\frac{1}{\eps^2} \cdot \|\b^*\|_0 \cdot  p$ term above comes from the first term in \eqref{eq:n_eps-app}, which, roughly speaking, reflects the error incurred by classical group lasso. Thus, the linear dependence of the sample complexity on $p$ appears to be a fundamental characteristic of the problem for GFF based random designs, and is inherent in both classical group lasso and the present heat flow based methods.

\subsection{Sample complexity and prediction guarantees for covariate structures based on Stochastic Block Models} \label{sec:recovery_block} 

\subsubsection{Spectral bounds}

We discuss herein upper and lower spectral bounds on block matrices. This is encapsulated in the following lemma.
\begin{lemma} \label{lem:block-alg_dec}
    We have 
    \begin{equation} \label{eq:block-alg_dec}
        \S = \l(1 - (a-b) \r) I_V + (a-b) \sum_{i=1}^k \ind_{\calC_i} \ind_{\calC_i}^T + b \ind_V \ind_V^T.
    \end{equation}
    Further, we have the bounds 
    \begin{equation} \label{eq:block-bounds}
        (1-a+b) \le \s_{\min}(\S) \le \s_{\max}(\S) \le (1-a+b) + (a-b) \max_i{|\calC_i|} +  b|V|.
    \end{equation}
\end{lemma}

The standard  choice  for connection probabilities in the stochastic block model in the community detection literature entails that $a=\tilde{a}/|V|$ and $b=\tilde{b}/|V|$, with the ratio $\tilde{a}/\tilde{b}$ large but fixed; this is the setting we will work with in the present paper. We may obtain from Lemma \ref{lem:block-alg_dec} that $\s_{\max}(\S)$ and $\s_{\min}(\S)$ are both $O(1)$. We record this as \[ \rho(\S)=1; ~\s_{\max}(\S)= O(1)  \text{ and} ~\kappa_\S(s) \ge  \s_{\min}(\S) \gtrsim O(1). \]

Combining these bounds with \eqref{eq:n_eps-app}, and working in the setting of typical clustered networks which entails balanced group sizes (so that we can approximate $s |\calC_{\max}| \simeq \|\b^*\|_0$), we  obtain the bound 
\[ 
    \l( \frac{n}{\log n} \r)_{\text{SBM}} \gtrsim  \max \l\{ \frac{1}{\eps^2} \cdot \|\b^*\|_0,~~ \|\b^*\|_0 \cdot \log p \r\}. 
\] 

\subsection{Proofs for spectral bounds}  \label{sec:spectral_proofs}

\subsubsection{Proofs of spectral bounds for GFF}

\begin{proof}[Proof of Lemma \ref{lem:graph_spectra}]
    The bound \eqref{eq:dense_eig_bound} has been established by \cite{Stanley}, which we refer the interested reader to for a detailed proof. 

    The key ingredient in \eqref{eq:sparse_eig_bound} is the basic spectral inequality on the adjacency matrix $A$ of a graph, given by $\s_{\max}(A) \le d_{\max}$.  This follows from the non-negative definiteness of the Laplacian: \[ 0 \preccurlyeq L = D - A \implies A \preccurlyeq D \implies \|A\|_{\mathrm{op}} \le \|D\|_{\mathrm{op}}  \implies \s_{\max}(A) \le \s_{\max}(D) = d_{\max}. \]
    As result, we may write
    \[ \s_{\max}(L) = \|L\|_{\mathrm{op}} = \|D-A\|_{\mathrm{op}} \le \|D\|_{\mathrm{op}} + \|A\|_{\mathrm{op}} \le 2d_{\max}. \]

    Finally, the spectral lower bound \eqref{eq:Cheeger_eig_bound} is one direction of the celebrated Cheeger's inequality on the smallest non-zero eigenvalue of the normalised graph Laplacian \citep{spielman2007spectral,spielman2012spectral}.
\end{proof}

\subsubsection{Proofs of spectral bounds for block models }

\begin{proof}[Proof of Lemma \ref{lem:block-alg_dec}.]
    The expression 
    \eqref{eq:block-alg_dec} follows from a simple algebraic decomposition of the matrix $\S=I_V + A$, which can be verified via a direct computation. 

    In order to obtain the singular value bounds on $\S$, we will make repeated use of the following spectral inequality for non-negative definite (\textit{abbrv.}  n.n.d.) matrices. Suppose 
    $\{B_i\}_{i=1}^{k_1}$ and $\{C_i\}_{i=1}^{k_2}$ are n.n.d. matrices such that for some n.n.d. matrix $M$ we have
    \begin{equation} \label{eq:eig_dom_1}
        \sum_{i=1}^{k_1} B_i  \preccurlyeq M \preccurlyeq \sum_{j=1}^{k_2} C_j,
    \end{equation}
    where $\preccurlyeq$ denotes smaller than or equal to in the n.n.d. order.
    Then we must have
    \begin{equation} \label{eq:eig_dom_2}
        \max_i \s_{\min}(B_i)  \le  \s_{\min}(M) \le \s_{\max}(M) \le \sum_{j=1}^{k_2} \s_{\max}(C_j).
    \end{equation}

    We now apply the \eqref{eq:eig_dom_2} to $M=\S$ with  $k_1=k_2=3;  ~B_{1}=C_{1}=\l(1 - (a-b) \r) I_V; ~B_2=C_2= (a-b) \sum_{i=1}^k \ind_{\calC_i} \ind_{\calC_i}^T $,  and $B_{3}=C_{3}=b \ind_V \ind_V^T$. 

    In order to deal with $\s_{\max}(B_2)$, it remains to observe that for any subset $S \subseteq V$ and any non-negative scalar $c$, we have $\s_{\min}(c \ind_S \ind_S^T)=0$ and $\s_{\max}(c \ind_S \ind_S^T) = c |S|$, and for two such \textit{disjoint} subsets $S_1,S_2 \subset V$, we have 
    \[   \s_{\max}( c \ind_{S_1} \ind_{S_1}^T + c \ind_{S_1} \ind_{S_1}^T)  \le c \max\{ \s_{\max}(\ind_{S_1} \ind_{S_1}^T, \ind_{S_2} \ind_{S_2}^T)  \} . \]
\end{proof}

\section{Further experimental results}
\subsection*{Estimation in GFF-s}
When the dimensionality $p$ is close to or more than the sample size $n$, estimation of the graph appears to become difficult, and in fact becomes progressively harder as the dimensionality increases. In a simulation experiment, we sampled observations from a GFF on a graph on $p = 200$ vertices drawn from the same SBM as described in the main text. Figure~\ref{fig:est-graph-GFF} shows the true adjacency matrix and the estimated one from thresholding correlation matrices.  

\begin{figure}[!ht]
    \centering
    \begin{tabular}{ccc}
        \includegraphics[scale = 0.2]{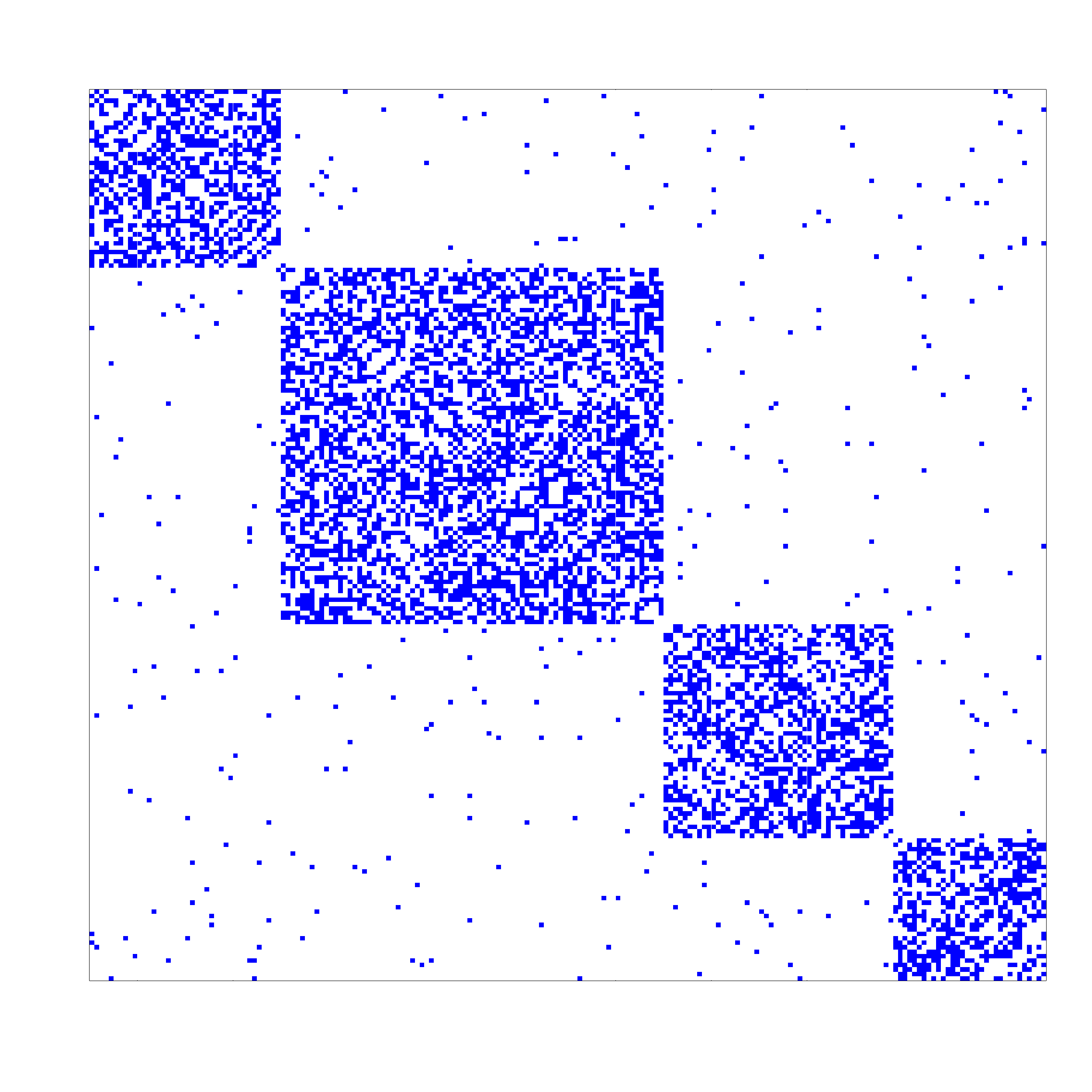} & \includegraphics[scale = 0.2]{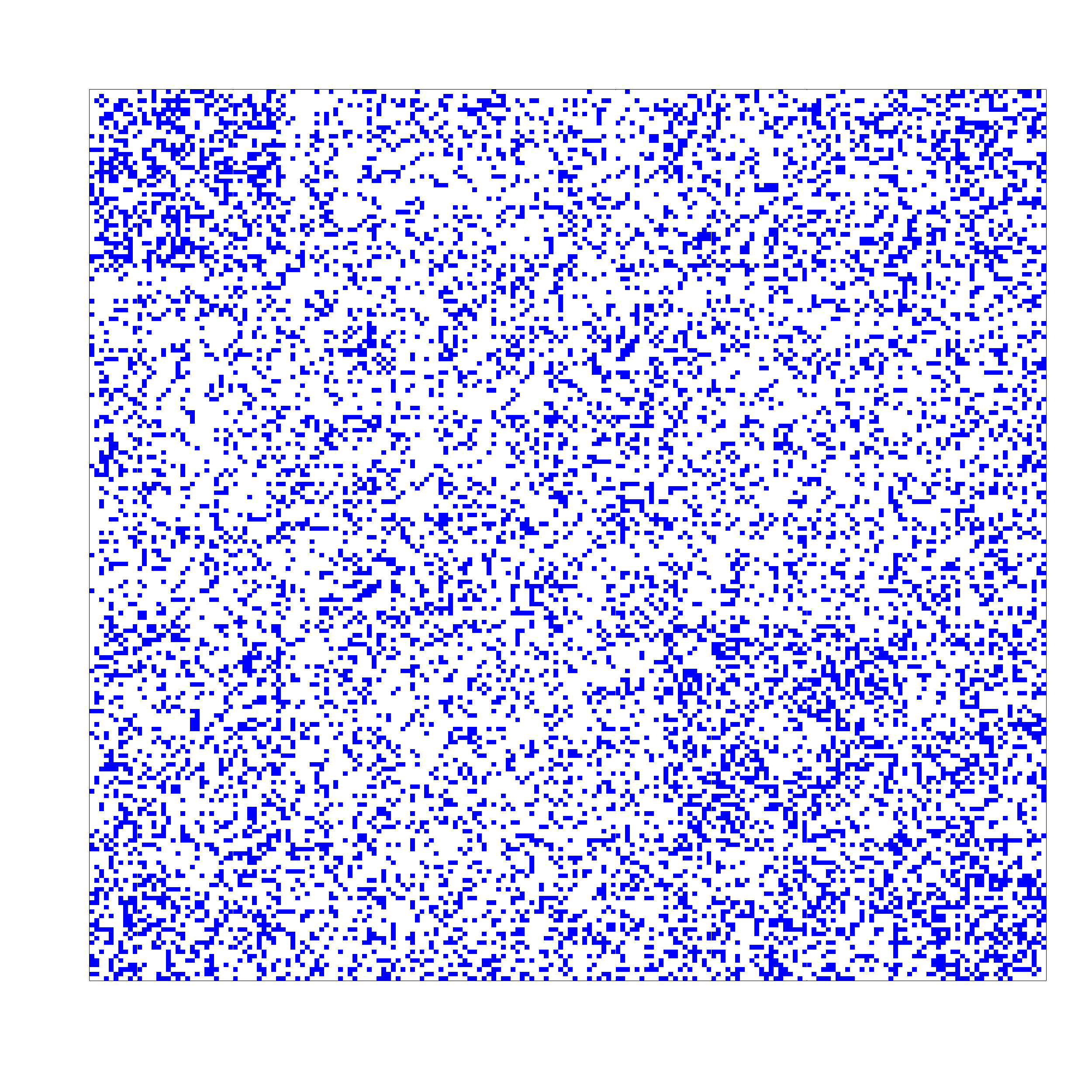} & \includegraphics[scale = 0.2]{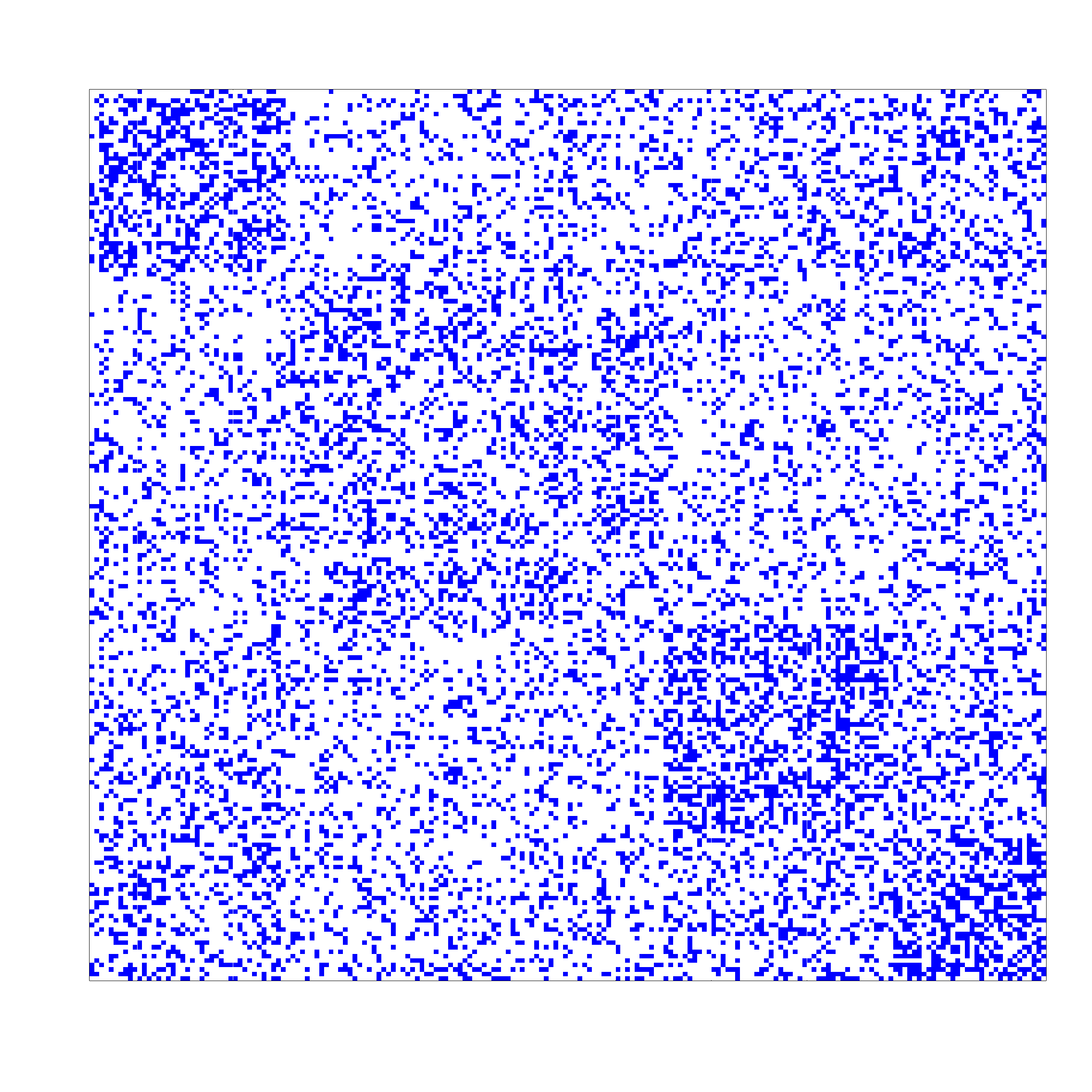} \\
        Truth & $n = 200$ & $n = 400$ \\
        \includegraphics[scale = 0.2]{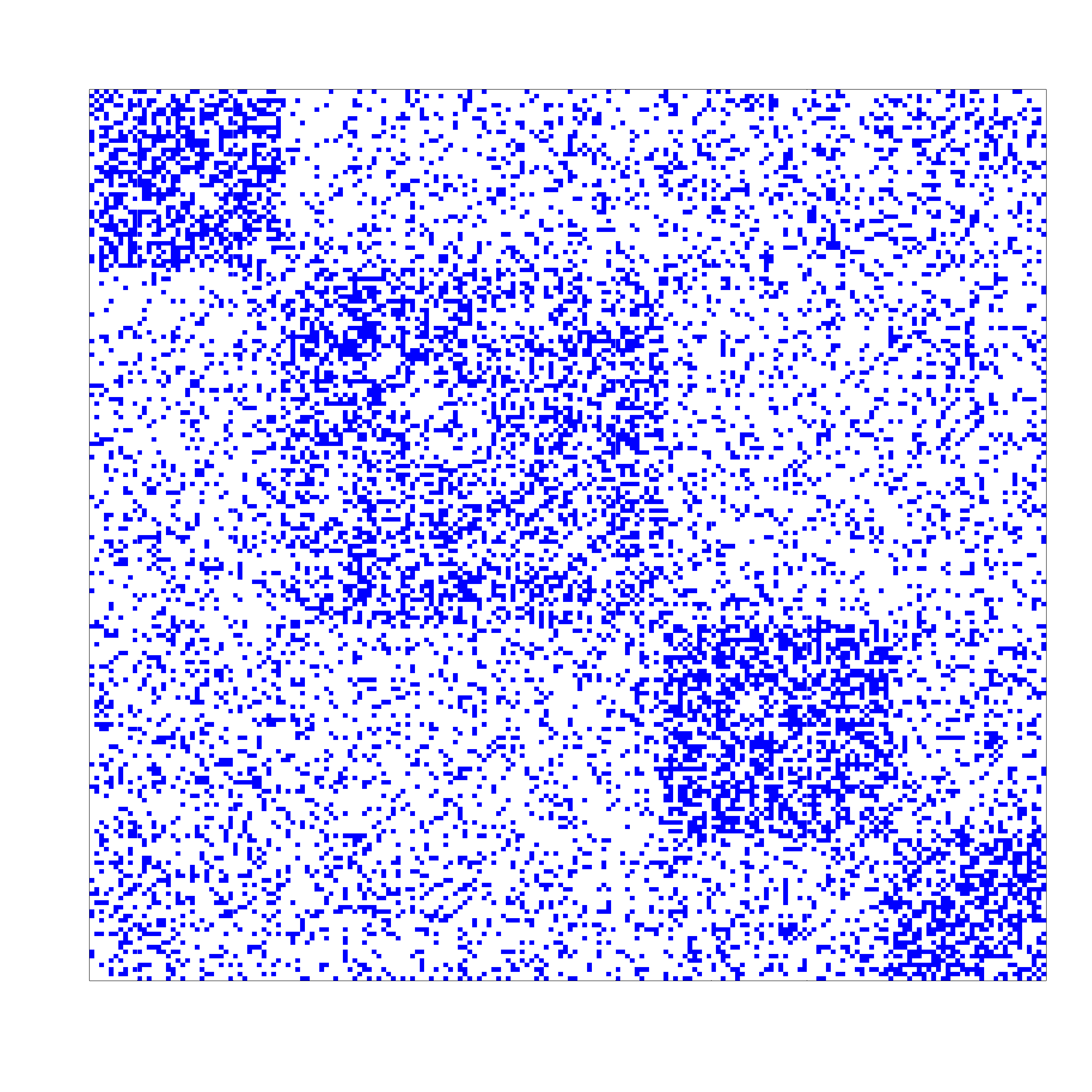} & \includegraphics[scale = 0.2]{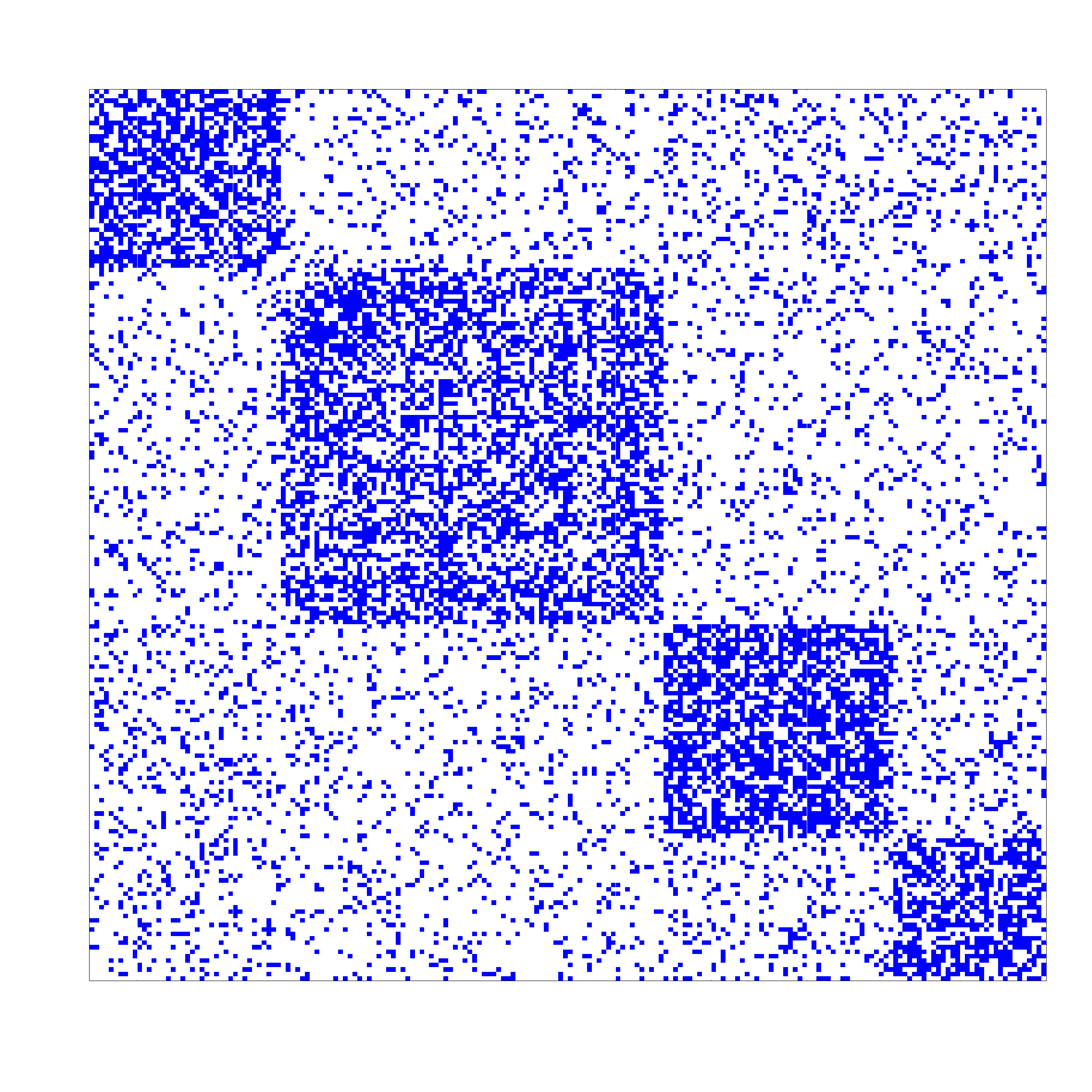} & \includegraphics[scale = 0.2]{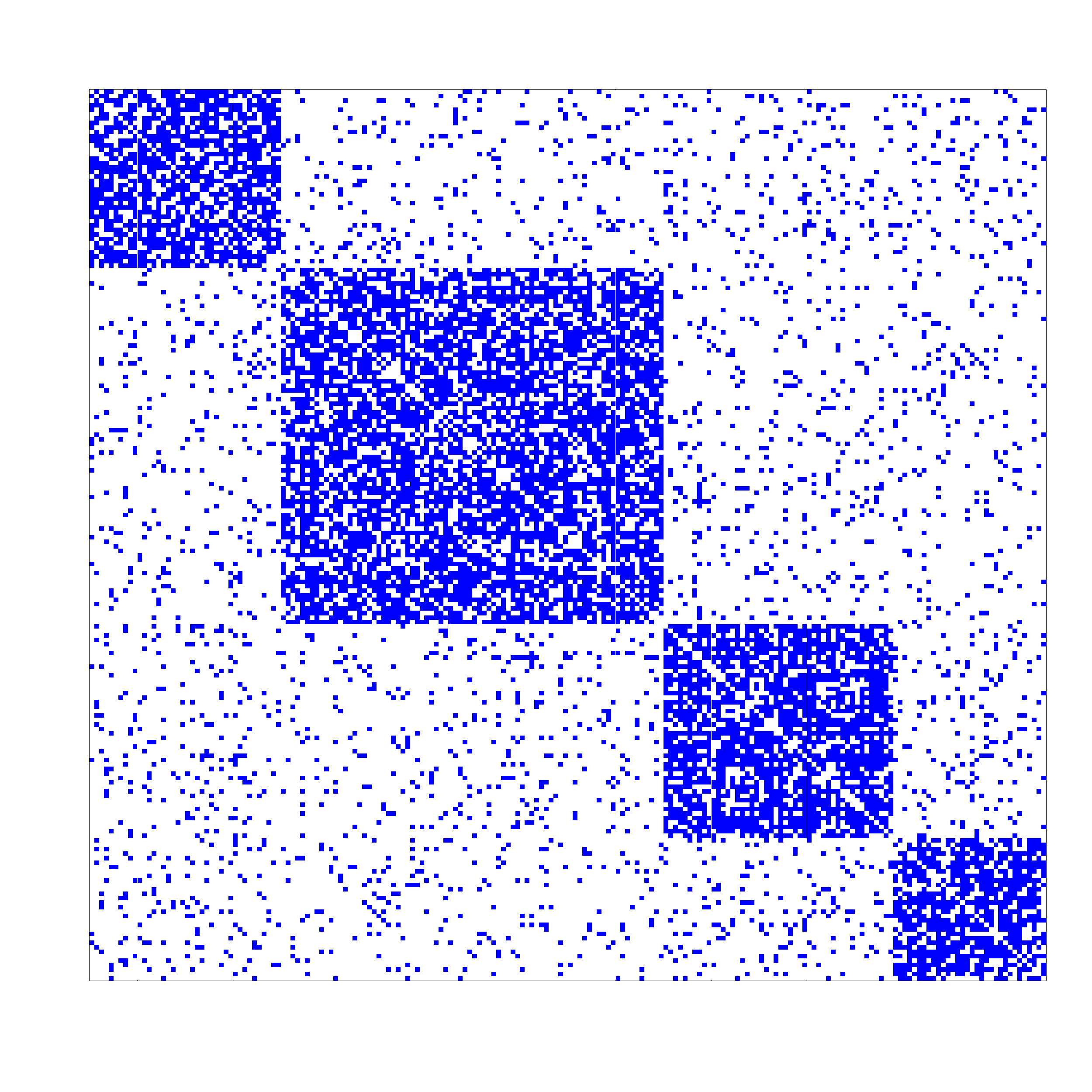} \\
        $n = 800$ & $n = 1600$ & $n = 3200$
    \end{tabular}
    \caption{Graph estimated (using Algorithm~\ref{alg:graph_estimation}) from a sample of size $n \in \{200, 400, 800, 1600, 3200\}$ from a GFF on a graph on $p = 200$ vertices generated from a stochastic block model with parameters $a = 0.5$, $b = 0.01$. For estimating the covariance matrix $\Sigma$, we use the shrinkage estimator of \cite{chen2010shrinkage}. We also see similar results when using the sample covariance matrix or the graphical lasso \citep{friedman2008sparse} for estimating $\Sigma$.}
    \label{fig:est-graph-GFF}
\end{figure}

\end{document}